\documentclass[aps, pra, twocolumn, groupedaddress, amsmath, amssymb, floatfix]{revtex4-2}
\usepackage{amsthm}
\usepackage{booktabs}
\usepackage{hyperref}
\usepackage{graphicx}
\usepackage{amssymb}
\usepackage{tikz}
\usetikzlibrary{calc, positioning, shapes.geometric, arrows.meta, backgrounds, fit}
\usepackage{pgfplots}
\pgfplotsset{compat=1.18}
\usepackage{xcolor}
\hypersetup{
    colorlinks=true,
    linkcolor=blue!70!black,
    citecolor=green!60!black,
    urlcolor=blue!70!black,
    bookmarksnumbered=true
}
\theoremstyle{plain}
\newtheorem{theorem}{Theorem}

\theoremstyle{remark}
\newtheorem{remark}{Remark}

\DeclareMathOperator{\tr}{tr}

\begin{document}

\title{Homomorphic Aggregation of Continuous-Variable GKP States}

\author{Nilesh Vyas}
\affiliation{Airbus Central R\&T, Taufkirchen, 82024 Germany}
\date{\today}

\begin{abstract}
Aggregating logical information in continuous-variable quantum networks is essential for distributed quantum architecture. However, direct passive linear optics degrade non-Gaussian Gottesman-Kitaev-Preskill (GKP) grid states via symplectic lattice compression and entanglement-induced decoherence when auxiliary modes are discarded. We present an active, measurement-based continuous-variable network primitive for combining spatially distributed computational-basis payloads. Utilizing GKP Bell states, homodyne measurements, and conditional feed-forward phase-space displacements, we construct a completely positive trace-preserving (CPTP) map that evaluates a logical XOR operation on computational-basis inputs. We bound the Heisenberg action of the physical finite-squeezing channel on logical Pauli generators, demonstrate measurement-specific homodyne-outcome hiding for continuous-variable quantum one-time pads (CV-OTP) with an explicit prefactor bound $D_{\mathrm{TV}} \le 1.60 e^{-r}$, and evaluate logical success probabilities under physical optical loss and network scaling constraints.
\end{abstract}

\maketitle

\section{Introduction}
\label{sec:intro}

Continuous-variable (CV) quantum networks and distributed quantum computation require robust encoding schemes to protect discrete logical information within infinite-dimensional Hilbert spaces. The Gottesman-Kitaev-Preskill (GKP) code \cite{Gottesman2001} is a primary candidate for fault-tolerant encoding, utilizing non-local superpositions of position and momentum eigenstates to form a stabilizing phase-space lattice. This translational symmetry intrinsically protects quantum data against low-weight displacement noise, making it highly attractive for near-term scalable quantum hardware \cite{Terhal2020, CampagneIbarcq2020, Baragiola2019}.

%A significant challenge in scaling multi-node quantum networks is the homomorphic aggregation of spatially distributed, encrypted states within continuous-variable architectures \cite{Larsen2021}. In the standard Gaussian coherent-state paradigm, encrypted inputs can be aggregated via a passive linear optical network. For two spatial modes, $1$ and $2$, a 50:50 beam-splitter unitary $\hat{B}_{12}$ acts linearly on coherent states: $\hat{B}_{12}(|\alpha_1\rangle_1 \otimes |\alpha_2\rangle_2) = |(\alpha_1 + \alpha_2)/\sqrt{2}\rangle_1 \otimes |(\alpha_1 - \alpha_2)/\sqrt{2}\rangle_2$. Because coherent states do not entangle under this quadratic interaction, applying a partial trace over the difference mode isolates a pure, aggregated coherent state. 

A major challenge in multi-node quantum networks is homomorphically aggregating spatially distributed states \cite{Larsen2021}. While passive linear optics cleanly aggregate Gaussian coherent states, where tracing out the difference mode yields a pure state, this fails for non-Gaussian encodings. Like rotation-symmetric codes (e.g., Cat or Binomial codes), GKP states become entangled under beam-splitter dynamics, leading to severe entropy injection upon tracing out spatial modes.

This passive approach encounters physical barriers when applied directly to discrete logical qubits encoded in grid states \cite{Bartlett2002}. A 50:50 beam splitter scales the phase-space lattice quadratures by $1/\sqrt{2}$, compressing the standard square-lattice GKP stabilizer generators' translational period from $2\sqrt{\pi}$ to $\sqrt{2\pi}$, thereby pushing the mapped quadratures outside the stabilized logical subspace. Furthermore, highly non-classical, non-Gaussian GKP states become rigorously entangled under beam-splitter dynamics; tracing out auxiliary spatial modes  injects entropy and decoheres the logical payload into a mixed state.

While continuous-variable cluster states and measurement-based quantum computing (MBQC) teleportation protocols are well-established for universal computation within centralized fault-tolerant quantum processors \cite{Alexander2014, Larsen2021}, the problem of routing and aggregating independently encrypted spatial modes across a distributed network requires a dedicated continuous-variable routing architecture. 

In this paper, we establish a theoretical framework for CV quantum payload aggregation that combines $n$ independent GKP computational-basis inputs while bypassing the limitations of direct passive linear optical mixing. Our specific contributions are fourfold: 

First, we  analyze the failure modes of direct passive spatial GKP aggregation, driven by deterministic symplectic lattice compression and entanglement-induced entropy injection when auxiliary modes are traced out. Second, we construct an active, measurement-based homomorphic router using a sequence of coherent CV SUM gates and homodyne detections, preserving the logical code-space geometry for computational-basis payloads. Third, we prove that this measurement channel acts as a measurement-specific homodyne-outcome hiding mechanism, masking the logical payload from intermediate routing nodes with an explicit bound $D_{\mathrm{TV}} \le 1.60 e^{-r}$. Finally, we analyze finite-squeezing noise and optical loss, demonstrating that for the sequential fusion architecture considered here, a logarithmic binary-tree routing topology minimizes circuit depth and suppresses cumulative temporal transmission loss.

\section{Context and Related Work}
\label{sec:context}

To contextualize our contribution, we  distinguish the continuous-variable (CV) network primitive proposed here from existing paradigms in both classical computer science and quantum information processing.

\paragraph{Measurement-Based Quantum State Aggregation}
We define \textit{measurement-based quantum state aggregation} as a completely positive trace-preserving (CPTP) map $\mathcal{E}: \mathcal{B}(\mathcal{H}_1 \otimes \dots \otimes \mathcal{H}_n) \to \mathcal{B}(\mathcal{H}_{\text{out}})$ that maps $n$ spatially distributed, individually encrypted quantum states into a single target logical payload Hilbert space without intermediate decryption. Unlike classical state-space reduction or quantum multi-signature schemes \cite{Wei2024}, this framework realizes the spatial fusion of non-Gaussian continuous-variable states across independent network nodes.

\paragraph{Analogy to Lattice-Based Cryptographic Frameworks}
Structurally, our continuous-variable aggregation framework serves as a physical analogue to classical lattice-based homomorphic commitment schemes, such as the BDLOP framework \cite{BDLOP}. While classical lattice schemes encode discrete message vectors onto algebraic module lattices over finite rings ($\mathbb{Z}_q^d$) using additive bounded noise, GKP codes embed discrete logical qubits into continuous two-dimensional real phase space ($\mathbb{R}^2$) protected by translational lattice symmetries. A central challenge in both settings is noise propagation under homomorphic addition: direct sequential fusion of physical GKP states accumulates additive Gaussian displacement variance scaling as $\mathcal{O}(n)$, whereas hierarchical tree topologies restrict this growth to $\mathcal{O}(\log n)$. Controlling this variance growth mandates intermediate GKP error correction and soft-decision syndrome extraction at network vertices to periodically reset the physical noise floor.

\paragraph{Quantum Homomorphic Encryption (QHE)}
Quantum Homomorphic Encryption (QHE) and its continuous-variable extensions (CV-QHE) \cite{Broadbent2015, Marshall2016} enable computation on encrypted quantum data. Rigorous bounds dictate that information-theoretically secure, fully homomorphic quantum encryption schemes for arbitrary circuits are impossible without exponential communication overhead \cite{Lai2018}. While CV-QHE frameworks target generic circuit evaluation on centralized quantum processors, distributed spatial aggregation introduces unique symplectic constraints. Rather than claiming unconditional security against arbitrary quantum adversarial measurements, our protocol operates under a localized network threat model: it guarantees asymptotic homodyne-outcome hiding (Theorem~\ref{thm:crypto_hiding}) against intermediate routing nodes possessing only classical measurement records.

\paragraph{Distributed MBQC and Teleportation Architectures}
GKP grid states and CV cluster states are established resources for measurement-based quantum computing (MBQC) \cite{Menicucci2006, Alexander2014, Larsen2021}. While MBQC teleportation circuits are well understood, adapting them to non-destructive, spatially distributed GKP state aggregation requires resolving specific non-Gaussian failure modes and establishing analytical bounds on physical noise propagation and homodyne outcome leakage.

\section{Limitations of Direct Passive Linear Optics}
\label{sec:limitations}

We formulate the GKP code space to  analyze the failure modes of passive state aggregation. The ideal single-mode square-lattice GKP code is stabilized by the commutative displacement operators $\hat{S}_q = \exp(i 2\sqrt{\pi} \hat{p})$ and $\hat{S}_p = \exp(-i 2\sqrt{\pi} \hat{q})$, where the quadrature operators satisfy $[\hat{q}, \hat{p}] = i$ (setting $\hbar = 1$).

The ideal logical basis states, $|0\rangle_L$ and $|1\rangle_L$, are infinite-energy Dirac combs over position eigenstates $|q\rangle_q$:
\begin{equation}
|0\rangle_L = \sum_{s \in \mathbb{Z}} |2s\sqrt{\pi}\rangle_q, \quad |1\rangle_L = \sum_{s \in \mathbb{Z}} |(2s+1)\sqrt{\pi}\rangle_q.
\end{equation}
The physical inter-codeword distance is $d_{\mathrm{min}} = \sqrt{\pi}$, establishing a maximum correctable displacement error boundary (Voronoi cell radius) of $q_{\mathrm{thresh}} = d_{\mathrm{min}}/2 = \sqrt{\pi}/2 \approx 0.8862$.

\subsection{Symplectic Scaling Limits}
\label{subsec:symplectic_scaling}

Passive linear optical operations correspond to orthogonal symplectic transformations $\mathbf{S} \in \mathrm{Sp}(2n, \mathbb{R}) \cap \mathrm{O}(2n)$ \cite{Weedbrook2012}. For a two-mode 50:50 beam splitter, the symplectic transformation on the joint quadrature vector $\hat{\mathbf{x}} = (\hat{q}_1, \hat{p}_1, \hat{q}_2, \hat{p}_2)^T$ is given by:
\begin{equation}
\hat{\mathbf{x}}_{\mathrm{out}} = \mathbf{S}_{\mathrm{BS}} \hat{\mathbf{x}} = \frac{1}{\sqrt{2}} \begin{pmatrix} I_2 & I_2 \\ I_2 & -I_2 \end{pmatrix} \hat{\mathbf{x}},
\end{equation}
where $I_2$ is the $2 \times 2$ identity matrix. While the orthogonal matrix $\mathbf{S}_{\mathrm{BS}}$ globally preserves the $4$-dimensional Euclidean norm in phase space, extracting the aggregated quantum payload via spatial reduction to mode 1 transforms the single-mode quadratures according to $\hat{q}_{\mathrm{out},1} = (\hat{q}_1 + \hat{q}_2)/\sqrt{2}$ and $\hat{p}_{\mathrm{out},1} = (\hat{p}_1 + \hat{p}_2)/\sqrt{2}$.

Consequently, the effective local inter-codeword distance contracts by a factor of $1/\sqrt{2}$, reducing the post-interaction distance to $d_{\mathrm{min}}^{(1)} = \sqrt{\pi/2}$. This compresses the maximum correctable noise threshold to $q_{\mathrm{thresh}}^{(1)} = \sqrt{\pi}/(2\sqrt{2}) \approx 0.6266$. Because $q_{\mathrm{thresh}}^{(1)} < q_{\mathrm{thresh}} = \sqrt{\pi}/2$, direct passive mixing compresses the GKP stabilizer lattice, pushing the state outside the stabilized logical subspace \cite{Gottesman2001, Terhal2020}.

\subsection{Inter-Mode Entanglement and Entropy Injection}
\label{subsec:entropy_injection}

While spatial reduction via partial tracing preserves state purity for Gaussian coherent states, this structural preservation fails for non-Gaussian GKP codewords.

% Applying the beam-splitter unitary operator $\hat{B}_{12} = \exp[\frac{\pi}{4}(\hat{a}_1^\dagger \hat{a}_2 - \hat{a}_1 \hat{a}_2^\dagger)]$ to a separable product of GKP codewords $|\Psi_{\mathrm{in}}\rangle = |\psi_1\rangle_L \otimes |\psi_2\rangle_L$ yields a non-trivially entangled state $|\Psi_{12}\rangle$ \cite{Kim2002}. Expressed via the Schmidt decomposition, $|\Psi_{12}\rangle = \sum_{k} \lambda_k |u_k\rangle_1 \otimes |v_k\rangle_2$, the Schmidt rank  exceeds unity ($d > 1$). Tracing out mode 2 yields a reduced state $\rho_{\mathrm{out},1} = \tr_2(|\Psi_{12}\rangle\langle\Psi_{12}|)$ with non-zero von Neumann entropy $S(\rho_{\mathrm{out},1}) = -\tr(\rho_{\mathrm{out},1} \ln \rho_{\mathrm{out},1}) > 0$. This entropy injection reduces the logical state purity, $\tr(\rho_{\mathrm{out},1}^2) < 1$, physically corrupting the logical qubit and demonstrating the failure of direct passive linear-optical aggregation for the square-lattice GKP encoding considered here.

Applying a 50:50 beam splitter $\hat{B}_{12} = \exp[\frac{\pi}{4}(\hat{a}_1^\dagger \hat{a}_2 - \hat{a}_1 \hat{a}_2^\dagger)]$ to separable GKP input states $|\Psi_{\mathrm{in}}\rangle = |\psi_1\rangle_L \otimes |\psi_2\rangle_L$ generates inter-mode entanglement $|\Psi_{12}\rangle = \sum_{k} \lambda_k |u_k\rangle_1 \otimes |v_k\rangle_2$ with Schmidt rank $d > 1$ \cite{Kim2002}. Tracing out mode 2 produces a reduced state $\rho_{\mathrm{out},1} = \tr_2(|\Psi_{12}\rangle\langle\Psi_{12}|)$ with von Neumann entropy $S(\rho_{\mathrm{out},1}) = -\tr(\rho_{\mathrm{out},1} \ln \rho_{\mathrm{out},1}) > 0$ and purity $\tr(\rho_{\mathrm{out},1}^2) < 1$, demonstrating that passive optical mixing decoheres square-lattice GKP logical states.

\begin{figure*}[t]
\centering
\resizebox{\linewidth}{!}{%
\begin{tikzpicture}[
    scale=1.0, >=Stealth,
    qwire/.style={thick, color=blue!75!black},
    cwire/.style={double, double distance=1.4pt, thick, color=black},
    bs/.style={
        rectangle, draw=black, thick, fill=blue!12, inner sep=0pt,
        path picture={
            \draw[blue!80!black, line width=1.2pt] 
                (path picture bounding box.south west) -- (path picture bounding box.north east);
        }
    },
    meter/.style={
        rectangle, draw=black, thick, fill=white, minimum width=0.9cm, minimum height=0.7cm, inner sep=0pt
    },
    controlbox/.style={
        rectangle, draw=black, thick, fill=orange!12, rounded corners=4pt,
        minimum width=4.2cm, minimum height=1.4cm, align=center, font=\small
    },
    disp/.style={
        rectangle, draw=black, thick, fill=red!12, rounded corners=4pt,
        minimum width=2.6cm, minimum height=0.9cm, align=center, font=\small
    },
    stagebox/.style={
        draw=black!30, dashed, fill=black!2, rounded corners=6pt
    },
    stageheader/.style={
        font=\bfseries\sffamily\small, color=black!85
    }
]

    % STAGE 1
    \node (psi1) at (0.0, 4.5) {$|\psi_1\rangle_L$};
    \node (vacA) at (0.0, 3.0) {$|0\rangle_{\text{GKP}, A}$};
    \node (vacB) at (0.0, 1.5) {$|0\rangle_{\text{GKP}, B}$};
    \node (psi2) at (0.0, 0.0) {$|\psi_2\rangle_L$};
    
    \node[bs, minimum width=0.7cm, minimum height=2.1cm] (bs1) at (1.6, 2.25) {};
    \node[font=\scriptsize\bfseries, color=black!60, above=2pt of bs1] {BS};

    \draw[qwire] (vacA.east) -- ($(bs1.west)+(0, 0.75)$);
    \draw[qwire] (vacB.east) -- ($(bs1.west)+(0, -0.75)$);

    \node[font=\scriptsize, color=blue!80!black, fill=white, draw=black!20, inner sep=2pt, rounded corners=2pt] 
        at (2.5, 2.25) {$|\Phi^+\rangle_{L, AB}$};

    % STAGE 2 & 3
    \node[bs, minimum width=0.7cm, minimum height=2.1cm] (bs2) at (3.8, 3.75) {};
    \node[font=\scriptsize\bfseries, color=black!60, above=2pt of bs2] {BS};

    \coordinate (swapA) at (5.5, 3.0);
    \coordinate (swapB) at (5.5, 1.5);
    \draw[qwire, line width=1.2pt] (swapA) -- (swapB);
    
    \draw[thick, color=blue!75!black] ($(swapA)+(-3.5pt,-3.5pt)$) -- ($(swapA)+(3.5pt,3.5pt)$);
    \draw[thick, color=blue!75!black] ($(swapA)+(-3.5pt,3.5pt)$) -- ($(swapA)+(3.5pt,-3.5pt)$);
    
    \draw[thick, color=blue!75!black] ($(swapB)+(-3.5pt,-3.5pt)$) -- ($(swapB)+(3.5pt,3.5pt)$);
    \draw[thick, color=blue!75!black] ($(swapB)+(-3.5pt,3.5pt)$) -- ($(swapB)+(3.5pt,-3.5pt)$);
    
    \node[font=\tiny\bfseries, color=blue!80!black, above=3pt] at (swapA) {SWAP};

    \node[bs, minimum width=0.7cm, minimum height=2.1cm] (bs3) at (6.8, 0.75) {};
    \node[font=\scriptsize\bfseries, color=black!60, above=2pt of bs3] {BS};

    % HOMODYNE DETECTORS
    \node[meter] (m1) at (8.2, 4.5) {};
    \draw[thick] ($(m1)+(-0.25,-0.15)$) arc (180:0:0.25cm);
    \draw[thick, ->] ($(m1)+(0,-0.15)$) -- ($(m1)+(0.15,0.15)$);
    \node[font=\small, above=2pt of m1] {$\hat{q}_1$};

    \node[meter] (m2) at (8.2, 1.5) {};
    \draw[thick] ($(m2)+(-0.25,-0.15)$) arc (180:0:0.25cm);
    \draw[thick, ->] ($(m2)+(0,-0.15)$) -- ($(m2)+(0.15,0.15)$);
    \node[font=\small, above=2pt of m2] {$\hat{p}_A$};

    \node[meter] (m3) at (8.2, 0.0) {};
    \draw[thick] ($(m3)+(-0.25,-0.15)$) arc (180:0:0.25cm);
    \draw[thick, ->] ($(m3)+(0,-0.15)$) -- ($(m3)+(0.15,0.15)$);
    \node[font=\small, above=2pt of m3] {$\hat{q}_2$};

    \draw[qwire] (psi1.east) -- ($(bs2.west)+(0, 0.75)$);
    \draw[qwire] ($(bs1.east)+(0, 0.75)$) -- ($(bs2.west)+(0, -0.75)$);
    
    \draw[qwire] ($(bs2.east)+(0, 0.75)$) -- (m1.west);
    \draw[qwire] ($(bs2.east)+(0, -0.75)$) -- (swapA);
    \draw[qwire] (swapB) -- ($(bs3.west)+(0, 0.75)$);
    
    \draw[qwire] (psi2.east) -- ($(bs3.west)+(0, -0.75)$);
    \draw[qwire] ($(bs3.east)+(0, 0.75)$) -- (m2.west);
    \draw[qwire] ($(bs3.east)+(0, -0.75)$) -- (m3.west);

    \draw[qwire] ($(bs1.east)+(0, -0.75)$) -- (swapB);
    \draw[qwire] (swapA) -- (12.2, 3.0);

    % STAGE 4
    \node[controlbox] (ctrl) at (13.8, 0.75) {\textbf{Classical Control}\\[2pt]$\mathbf{m}_\Sigma = f(m_1, m_A, m_2)$};

    \node[disp] (disp) at (13.8, 3.0) {$\hat{D}_B(-\mathbf{m}_\Sigma)$};

    \draw[qwire] (12.2, 3.0) -- (disp.west);

    \draw[cwire] (m1.east) -- (10.0, 4.5) |- ([yshift=14pt]ctrl.west);
    \node[font=\small, fill=white, inner sep=2pt, draw=black!20, rounded corners=2pt] at (9.4, 4.5) {$m_1$};

    \draw[cwire] (m2.east) -- (10.4, 1.5) |- ([yshift=0pt]ctrl.west);
    \node[font=\small, fill=white, inner sep=2pt, draw=black!20, rounded corners=2pt] at (9.8, 1.5) {$m_A$};

    \draw[cwire] (m3.east) -- (10.8, 0.0) |- ([yshift=-14pt]ctrl.west);
    \node[font=\small, fill=white, inner sep=2pt, draw=black!20, rounded corners=2pt] at (10.2, 0.0) {$m_2$};

    \draw[cwire, -{Implies}] (ctrl.north) -- (disp.south);

    \draw[qwire, ->] (disp.east) -- (16.5, 3.0) 
        node[anchor=west] {$|\psi_1 \oplus \psi_2\rangle_L$};

    % BOUNDING BOXES
    \begin{scope}[on background layer]
        \node[stagebox, fit={(-0.5, -0.7) (2.8, 5.3)}] (st1) {};
        \node[stagebox, fit={(3.0, -0.7) (9.0, 5.3)}] (st2) {};
        \node[stagebox, fit={(9.2, -0.7) (18.2, 5.3)}] (st3) {};
    \end{scope}

    \node[stageheader] at (1.15, 5.6) {1. Resource Generation};
    \node[stageheader] at (6.00, 5.6) {2 \& 3. Cluster Formation \& CV-BSM};
    \node[stageheader] at (13.70, 5.6) {4. Classical Control \& Output};

\end{tikzpicture}%
}
\caption{Continuous-variable measurement-based quantum state aggregation protocol for $n=2$ logical inputs. Stage 1 generates a Gottesman-Kitaev-Preskill (GKP) Bell pair $|\Phi^+\rangle_{L, AB}$ by interfering two single-mode GKP qunaught states $|0\rangle_{\text{GKP}, A}$ and $|0\rangle_{\text{GKP}, B}$ at a 50:50 beam splitter (BS). In Stages 2 and 3, Data Mode 1 ($|\mu_1\rangle_L$) and Data Mode 2 ($|\mu_2\rangle_L$) couple to Ancilla Mode A via sequential beam splitters and a spatial SWAP gate, followed by homodyne detection of quadratures $\hat{q}_1, \hat{p}_A$, and $\hat{q}_2$. Stage 4 computes the net classical feed-forward displacement vector $\mathbf{m}_\Sigma = f(m_1, m_A, m_2)$, applying $\hat{D}_B(-\mathbf{m}_\Sigma)$ to unmeasured Ancilla Mode B to recover the aggregated logical state $|\mu_1 \oplus \mu_2\rangle_L$.}
\label{fig:gkp_homomorphic_aggregation}
\end{figure*}
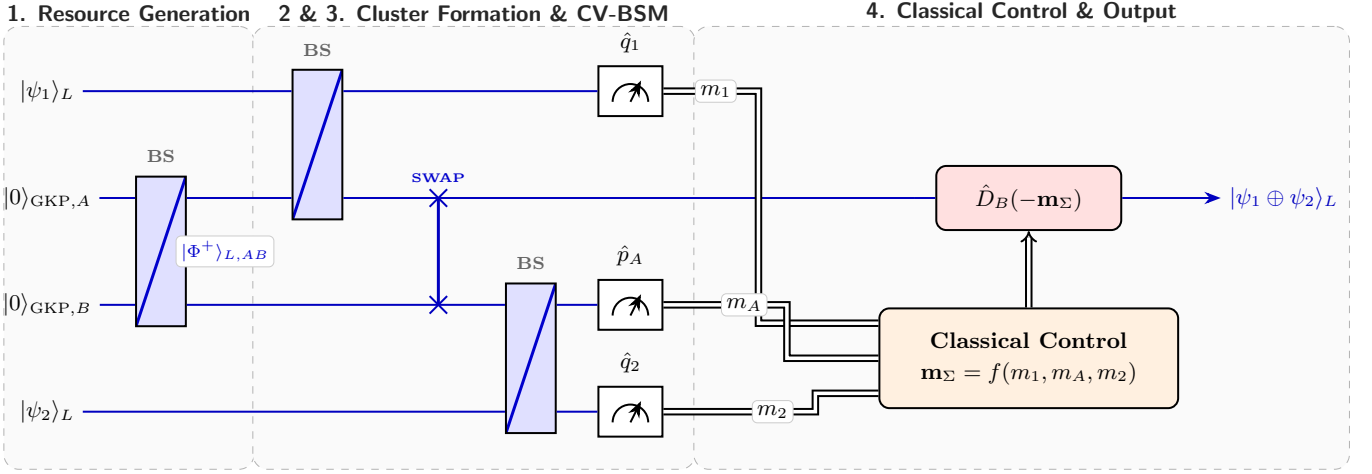

\section{Quantum Homomorphic Aggregation via MBQC Teleportation}
\label{sec:mbqc_aggregation}

To achieve homomorphic state aggregation without destroying the non-Gaussian structure of the GKP code, we construct the network primitive using measurement-based quantum computing (MBQC) teleportation circuits.

\subsection{The Ideal Continuous-Variable SUM Gate}
\label{subsec:sum_gate}
In continuous-variable quantum information, the foundational two-mode entangling operation is the Gaussian $\mathrm{SUM}$ gate, $\hat{C}_X = \exp(-i\hat{q}_1\hat{p}_2)$, serving as the physical analogue to the discrete logical $\mathrm{CNOT}$ gate. Under ideal unitary operation, this gate maps two input modes to two output modes, generating the following Heisenberg evolution:
\begin{align}
    \hat{q}_1' = \hat{q}_1, &\quad \hat{p}_1' = \hat{p}_1 - \hat{p}_2, \\
    \hat{q}_2' = \hat{q}_1 + \hat{q}_2, &\quad \hat{p}_2' = \hat{p}_2.
\end{align}
Because this continuous phase-space translation  aligns with the periodic symmetry of the GKP code space, the physical addition executes an  modulo-2 logic gate on the discrete payload. Specifically, the unitary evaluates to $\hat{C}_X(|\mu_1\rangle_L \otimes |\mu_2\rangle_L) = |\mu_1\rangle_L \otimes |\mu_1 \oplus \mu_2\rangle_L$. While this ideal unitary operation is a 2-to-2 mapping that preserves the control mode, our distributed network protocol realizes a destructive 2-to-1 spatial reduction that evaluates the  logical XOR in the ideal computational-basis model.

\subsection{Distributed Teleportative Fusion Protocol}
\label{subsec:protocol}
In a spatially distributed network, direct Hamiltonian coupling of remote modes is impossible. Furthermore, network routing requires aggregating distributed payloads into a single target node while  destructing the intermediate source modes. We achieve this homomorphic aggregation via the following measurement-based continuous-variable protocol (visualized in Fig.~\ref{fig:gkp_homomorphic_aggregation}):

\begin{enumerate}
    \item \textbf{Source Encryption (CV-OTP):} Prior to network transmission, the independent source nodes protect their discrete logical inputs ($|\mu_1\rangle_L$ and $|\mu_2\rangle_L$) by applying a continuous-variable one-time pad. Each state undergoes a random phase-space translation dictated by a uniformly sampled, independent continuous displacement mask (e.g., $q_{\mathrm{enc}} \sim \mathcal{U}[-\sqrt{\pi}, \sqrt{\pi})$).
    \item \textbf{Resource Generation:} The network prepares an offline, finite-energy continuous-variable EPR resource state (a GKP Bell pair) across Ancilla Modes A and B. This establishes the initial quadrature correlations $\hat{q}_A - \hat{q}_B = \hat{\xi}_{q,\mathrm{EPR}}$ and $\hat{p}_A + \hat{p}_B = \hat{\xi}_{p,\mathrm{EPR}}$.
    \item \textbf{Cluster Formation:} The encrypted data modes route to the intermediate vertex. They sequentially interfere with Ancilla Mode A via 50:50 beam splitters, entangling the masked input payloads with the ancilla rail.
    \item \textbf{Measurement (CV-BSM):} The vertex executes destructive homodyne detection on the three interacting output ports ($\hat{q}_1$, $\hat{p}_A$, $\hat{q}_2$). This continuous-variable Bell-state measurement (CV-BSM) extracts a set of classical, continuous quadrature syndromes.
    \item \textbf{Classical Feed-Forward:} A classical controller computes the linear algebraic combination of the measured syndromes to form a collective displacement vector $\mathbf{m}_\Sigma$. The target node holding the unmeasured Ancilla Mode B applies the local conditional displacement $\hat{D}_B(-\mathbf{m}_\Sigma)$.
\end{enumerate}

This sequence destructively consumes the input modes while teleporting the combined logical payload $|\mu_1 \oplus \mu_2\rangle_L$ onto the unmeasured target rail.

\subsection{Logical Action on the Target GKP Lattice}
\label{subsec:logical_action}
To demonstrate that this continuous-variable interaction executes an  homomorphic addition, we map the Heisenberg evolution of the position operators onto the discrete GKP lattice support. As derived in Appendix~\ref{app:heisenberg_evolution}, the teleportation protocol maps the target position operator to:
\begin{equation}
    \hat{q}_{\mathrm{target}} = \hat{q}_1 + \hat{q}_2 - \hat{\xi}_{q,\mathrm{EPR}}.
\end{equation}

For ideal computational basis inputs $|\mu_1\rangle_L$ and $|\mu_2\rangle_L$, the probability amplitudes are supported  on position eigenstates corresponding to the sub-lattices $q_j = (2s_j + \mu_j)\sqrt{\pi}$, where $s_j \in \mathbb{Z}$ and $\mu_j \in \{0, 1\}$. Evaluating the sum of the ideal lattice supports yields:
\begin{equation}
    q_{\mathrm{ideal}} = (2s_1 + \mu_1)\sqrt{\pi} + (2s_2 + \mu_2)\sqrt{\pi}.
\end{equation}

In standard binary arithmetic, the sum of two logical bits decomposes into a carry bit $c \in \{0, 1\}$ and a modulo-2 sum: $\mu_1 + \mu_2 = 2c + (\mu_1 \oplus \mu_2)$. Substituting this expansion isolates the logical payload from the lattice translation:
\begin{equation}
    q_{\mathrm{ideal}} = \left(2(s_1 + s_2 + c) + \mu_1 \oplus \mu_2\right)\sqrt{\pi}.
\end{equation}
Defining a new lattice integer $S = s_1 + s_2 + c$, the ideal support simplifies to $q_{\mathrm{ideal}} = (2S + \mu_1 \oplus \mu_2)\sqrt{\pi}$. Therefore, the continuous physical addition of position operators inherently maps the state onto the  GKP sub-lattice corresponding to the modulo-2 logical sum $|\mu_1 \oplus \mu_2\rangle_L$.

In the physical finite-energy channel $\mathcal{E}_{\mathrm{SUM}}$, the residual quantum noise $\hat{\xi}_{q,\mathrm{EPR}}$ shifts the target state off the ideal lattice points:
\begin{equation}
    \hat{q}_{\mathrm{target}} = (2S + \mu_1 \oplus \mu_2)\sqrt{\pi} - \hat{\xi}_{q,\mathrm{EPR}}.
\end{equation}
The logical operation avoids a discrete Pauli fault if and only if this continuous displacement noise remains bounded within the correctable GKP Voronoi cell, $|\hat{\xi}_{q,\mathrm{EPR}}| < \sqrt{\pi}/2$.

\subsection{Decoherence Limits and Classical Payload Constraints}
\label{subsec:decoherence}

The 2-to-1 spatial aggregation map reduces the bipartite input space $\mathcal{H}_1 \otimes \mathcal{H}_2$ to a single output mode $\mathcal{H}_{\mathrm{target}}$ via a partial trace over the control register. While tracing out a subsystem of an entangled state reduces state purity, restricting the input domain to computational-basis states prevents inter-mode entanglement generation during the SUM operation, preserving state purity across the spatial reduction.

However, the protocol circumvents this decoherence penalty because it is scoped as a quantum-secure homomorphic evaluator for computational-basis payloads. By restricting inputs to discrete computational basis states, the protocol avoids generating the entanglement that leads to mixedness during the partial trace. We  emphasize: the 2-to-1 map computes the logical XOR of computational-basis GKP states; it does not preserve arbitrary quantum superpositions when one register is discarded.

To formalize this logical purity preservation, we evaluate the ideal continuous-variable mapping. By temporarily isolating the ideal channel from the finite-energy hardware constraints, we mathematically decouple the fundamental informational entropy (induced by the spatial partial trace) from the physical continuous-variable entropy (induced by finite optical squeezing, subsequently evaluated in Sec.~\ref{sec:channel_integrity}).

Under ideal operation, the MBQC protocol executes an effective logical SUM gate ($\hat{C}_X$) across the input modes, mapping $|\mu_1\rangle_L \otimes |\mu_2\rangle_L \to |\mu_1\rangle_L \otimes |\mu_1 \oplus \mu_2\rangle_L$. The 2-to-1 spatial reduction corresponds mathematically to a partial trace over the effective control mode:
\begin{equation}
\mathcal{E}_{\mathrm{agg}}(\rho_{\mathrm{in}}) = \mathrm{Tr}_1 \left[ \hat{C}_X \rho_{\mathrm{in}} \hat{C}_X^\dagger \right].
\end{equation}

For a generic separable input superposition, $|\Psi_{\mathrm{in}}\rangle_L = (\alpha|0\rangle_L + \beta|1\rangle_L) \otimes |\mu_2\rangle_L$, the $\hat{C}_X$ operator generates a logically entangled state. Tracing out Mode 1 annihilates the off-diagonal coherences ($\alpha\beta^*$), reducing the target to a mixed state with  positive von Neumann entropy:
\begin{equation}
\mathcal{E}_{\mathrm{agg}}(|\Psi_{\mathrm{in}}\rangle\langle\Psi_{\mathrm{in}}|) = |\alpha|^2 |\mu_2\rangle_L\langle\mu_2| + |\beta|^2 |\bar{\mu}_2\rangle_L\langle\bar{\mu}_2|.
\end{equation}

Conversely, when the inputs are restricted to computational basis states, $\rho_{\mathrm{in}} = |\mu_1\rangle_L\langle\mu_1| \otimes |\mu_2\rangle_L\langle\mu_2|$, the $\hat{C}_X$ gate evaluates deterministically without generating superposition-based entanglement. The partial trace isolates the target mode  without injecting entropy:
\begin{equation}
\mathcal{E}_{\mathrm{agg}}(|\mu_1\rangle_L\langle\mu_1| \otimes |\mu_2\rangle_L\langle\mu_2|) = |\mu_1 \oplus \mu_2\rangle_{L}\langle\mu_1 \oplus \mu_2|.
\end{equation}
Unlike passive linear optics (Sec.~\ref{subsec:entropy_injection}), this active MBQC topology structurally isolates and preserves the  logical purity of the computational basis during spatial reduction.

\section{Cryptographic Syndrome Hiding}
\label{sec:crypto_hiding}
In the measurement-based routing architecture, intermediate network nodes perform continuous-variable Bell-state measurements, extracting the raw classical measurement vector $\mathbf{m}$ to compute the feed-forward syndrome $\mathbf{m}_\Sigma$. The threat model assumes these intermediate routers are ``honest-but-curious'': they faithfully execute the homodyne measurements and classical feed-forward, but may attempt to deduce the logical payload from these classical measurement records.

Because the router accesses only the classical homodyne outcomes, the security requirement maps to \textit{measurement-specific homodyne-outcome hiding}. We prove that the probability distribution of the extracted syndrome is asymptotically independent of the encoded logical state.

As formalized in Step 1 of the routing protocol (Sec.~\ref{subsec:protocol}), the source nodes protect the input states prior to transmission using a continuous-variable one-time pad (CV-OTP). Consequently, the intermediate router is forced to evaluate the homomorphic sum on these uniformly masked quadratures without access to the decryption key.

\begin{theorem}[Measurement-Specific Homodyne Hiding]
\label{thm:crypto_hiding}
For a finite-energy GKP state with squeezing parameter $r$, homodyne syndrome extraction under a uniformly sampled continuous displacement mask $q_{\mathrm{enc}} \sim \mathcal{U}[-\sqrt{\pi}, \sqrt{\pi})$ achieves asymptotic homodyne-outcome hiding. The total variation distance between classical syndrome probability distributions for any orthogonal logical states is bounded by:
\begin{equation}
D_{\mathrm{TV}}(P_0, P_1) \le C e^{-r},
\end{equation}
where $C \le 1.60$ for all physical squeezing levels $r \ge 5.0~\mathrm{dB}$.
\end{theorem}

\begin{proof}
\textbf{Spatial Probability Distributions: }
Let $p_\mu(q)$ denote the position quadrature probability distribution of the GKP state for the logical codeword $\mu \in \{0, 1\}$. For an ideal, infinite-energy GKP state, this distribution is a Dirac comb \cite{Gottesman2001}:
\begin{equation}
p_\mu^{\mathrm{ideal}}(q) \propto \sum_{s \in \mathbb{Z}} \delta\left(q - (2s + \mu)\sqrt{\pi}\right).
\end{equation}
For a physical finite-energy state, this ideal comb is broadened by single-mode position quadrature variance $\sigma_0^2 = \frac{1}{2} e^{-2r}$ due to finite optical squeezing, and bounded by a macroscopic Gaussian envelope of variance $\Delta^2 = \frac{1}{2} e^{2r}$ \cite{Terhal2020}. In subsequent channel calculations, the combined two-mode variance of the EPR Bell state evaluates to $\sigma_{\mathrm{EPR}}^2 = 2\sigma_0^2 = e^{-2r}$.

To enact the CV-OTP, the classical encryption mask $q_{\mathrm{enc}}$ is drawn uniformly over one full logical lattice period $2\sqrt{\pi}$:
\begin{equation}
f_U(q) = \frac{1}{2\sqrt{\pi}} \quad \text{for } q \in [-\sqrt{\pi}, \sqrt{\pi}),
\end{equation}
and zero otherwise. Because the extracted classical observable $m$ is the sum of the intrinsic quadrature outcome $q$ and the independent classical mask $q_{\mathrm{enc}}$, the probability distribution of the encrypted syndrome, $P_\mu(m)$, evaluates to the continuous convolution of their respective distributions:
\begin{equation}
P_\mu(m) = (p_\mu * f_U)(m).
\end{equation}

\textbf{Transformation to the Fourier Domain: }
To evaluate the total variation distance $D_{\mathrm{TV}}(P_0, P_1) = \frac{1}{2} \int_{\mathbb{R}} |P_0(m) - P_1(m)| \, dm$, we map the probability distributions to the spatial frequency domain using the characteristic function $\tilde{F}(k) = \int_{\mathbb{R}} F(q) e^{-ikq} \, dq$. According to the Convolution Theorem, the Fourier transform of a real-space convolution is the pointwise product of their individual Fourier transforms. This mathematically translates the physical displacement encryption into a spectral filtering operation:
\begin{equation}
\tilde{P}_\mu(k) = \tilde{p}_\mu(k) \cdot \tilde{f}_U(k).
\end{equation}

We evaluate the spectral filter corresponding to the uniform mask $f_U(q)$:
\begin{equation}
\tilde{f}_U(k) = \frac{1}{2\sqrt{\pi}} \int_{-\sqrt{\pi}}^{\sqrt{\pi}} e^{-ikq} \, dq = \frac{\sin(k\sqrt{\pi})}{k\sqrt{\pi}} \equiv \mathrm{sinc}(k\sqrt{\pi}).
\end{equation}
This sinc function possesses  mathematical roots at $k_n = n\sqrt{\pi}$ for all non-zero integers $n \in \mathbb{Z} \setminus \{0\}$.

Next, we derive the characteristic function of the ideal GKP state. Taking the Fourier transform of the shifted Dirac comb $p_\mu^{\mathrm{ideal}}(q)$ yields:
\begin{align}
\tilde{p}_\mu^{\mathrm{ideal}}(k) &\propto \int_{-\infty}^{\infty} \sum_{s \in \mathbb{Z}} \delta\left(q - (2s + \mu)\sqrt{\pi}\right) e^{-ikq} \, dq \nonumber \\
&= \sum_{s \in \mathbb{Z}} e^{-ik(2s + \mu)\sqrt{\pi}} \nonumber \\
&= e^{-i \mu k \sqrt{\pi}} \sum_{s \in \mathbb{Z}} e^{-i s (2k\sqrt{\pi})}.
\end{align}
We apply the Poisson Summation Formula, which dictates that $\sum_{s} e^{-i s (2\pi x)} \propto \sum_{n} \delta(x - n)$. Substituting $x = k/\sqrt{\pi}$ transforms the infinite exponential series into a periodic frequency comb bounded by delta functions at $k = n\sqrt{\pi}$. Evaluating the translation phase prefactor $e^{-i \mu k \sqrt{\pi}}$  at these delta supports ($k \to n\sqrt{\pi}$) yields $e^{-i \mu n \pi} = (-1)^{\mu n}$. The ideal spectral distribution thus evaluates  to:
\begin{equation}
\tilde{p}_\mu^{\mathrm{ideal}}(k) \propto \sum_{n \in \mathbb{Z}} (-1)^{\mu n} \delta(k - n\sqrt{\pi}).
\end{equation}

By taking the difference between the two computational basis states ($\mu = 0$ and $\mu = 1$), we isolate the spatial frequencies that carry logical distinguishability:
\begin{equation}
\tilde{p}_0^{\mathrm{ideal}}(k) - \tilde{p}_1^{\mathrm{ideal}}(k) \propto \sum_{n \in \mathbb{Z}} \left[ 1 - (-1)^n \right] \delta(k - n\sqrt{\pi}).
\end{equation}
For even integers $n$, the term $[1 - (-1)^n]$ vanishes. The logical distinguishability is therefore  constrained to the odd harmonic frequencies $k_n = (2n+1)\sqrt{\pi}$. 

Because the spectral filter of the uniform encryption mask evaluates to zero at all non-zero integer harmonics ($\tilde{f}_U(k_n) = 0$), the pointwise product  annihilates the distinguishability peaks. Thus, $\tilde{P}_0(k) - \tilde{P}_1(k) = 0$  in the infinite-squeezing limit,  yielding $D_{\mathrm{TV}}(P_0, P_1) = 0$.

\textbf{Bounding Finite-Energy Leakage via Real-Space $L^1$ Integration: }
For physical finite-energy GKP states, the ideal Dirac comb in the Fourier domain acquires an overall Gaussian envelope $\exp(-\sigma^2 k^2 / 2)$, while the individual harmonics broaden into localized Gaussian peaks with spectral variance $\sigma_k^2 = 1/\Delta^2 \approx 2e^{-2r}$.

Since $\tilde{f}_U(k_n) = 0$ at all odd harmonics, distinguishability is governed by the finite spectral width $\sigma_k$ of the broadened peaks. Expanding $\tilde{f}_U(k)$ to first order around $k_n = n\sqrt{\pi}$ for odd $n$ with detuning $\delta k$ gives:
%Because the uniform mask evaluates  to zero at the odd harmonics ($\tilde{f}_U(k_n) = 0$), the surviving informational leakage is driven entirely by the local gradient of the mask interacting with the finite width ($\sigma_k$) of these spectral peaks. Expanding $\tilde{f}_U(k)$ to first order around an odd harmonic $k_n = n\sqrt{\pi}$ (for odd integer $n$) with local detuning $\delta k$ yields:
\begin{equation}
\tilde{f}_U(n\sqrt{\pi} + \delta k) \approx \frac{(-1)^n}{n\sqrt{\pi}} \delta k.
\end{equation}
The local difference signal near the $n$-th harmonic in the Fourier domain is therefore proportional to the derivative of the unnormalized spectral peak:
\begin{equation}
\Delta \tilde{P}_n(\delta k) \propto \frac{(-1)^n}{n\sqrt{\pi}} \delta k \cdot e^{-\frac{\delta k^2}{2\sigma_k^2}} e^{-\frac{n^2 \pi \sigma^2}{2}}.
\end{equation}
Taking the inverse Fourier transform maps this localized spectral derivative into the real-space difference signal:
\begin{equation}
\Delta P_n(m) \propto \frac{(-1)^n}{n\sqrt{\pi}} \sigma_k^3 m e^{-\frac{m^2 \sigma_k^2}{2}} e^{-\frac{n^2 \pi \sigma^2}{2}}.
\end{equation}
Evaluating the $L^1$ norm of this $n$-th harmonic leakage in real space integrates out the spatial dimension ($m$), yielding a distinguishability bound that scales  linearly with the spectral peak width $\sigma_k$:
\begin{equation}
\int_{-\infty}^{\infty} |\Delta P_n(m)| \, dm \propto \frac{\sigma_k}{|n|\sqrt{\pi}} e^{-\frac{n^2 \pi \sigma^2}{2}}.
\end{equation}
% Substituting $\sigma_k \approx \sqrt{2}e^{-r}$ and $\sigma^2 \approx e^{-2r}/2$ isolates the exponential scaling:
% \begin{equation}
%  D_{\mathrm{TV}}(P_0, P_1) \le \left[ \frac{\sqrt{2}}{\pi} \sum_{n \text{ odd}} \frac{1}{|n|} \exp\left(-\frac{n^2 \pi e^{-2r}}{4}\right) \right] e^{-r} \equiv C(r) e^{-r}.
% \end{equation}
% For all physical squeezing levels $r \ge 5.0~\mathrm{dB}$ ($r \ge 0.5756$), $e^{-2r} \le 0.316$. Evaluating the rapidly converging series over odd integers gives $C(r) \le 1.60$, establishing the explicit tight bound $D_{\mathrm{TV}}(P_0, P_1) \le 1.60 e^{-r}$.
% As $r \to \infty$, $D_{\mathrm{TV}} \to 0$ exponentially,  proving that intermediate routers extract zero logical information from the classical syndrome records in the high-squeezing limit.
Substituting $\sigma_k \approx \sqrt{2}e^{-r}$ and $\sigma^2 \approx e^{-2r}/2$ yields:
\begin{align}
D_{\mathrm{TV}}(P_0, P_1) &\le \Big[ \frac{\sqrt{2}}{\pi} \sum_{n \text{ odd}} \frac{1}{|n|}  \exp\left(-\frac{n^2 \pi e^{-2r}}{4}\right) \Big] e^{-r} \nonumber \\ 
&\equiv C(r) e^{-r}.
\end{align}
For all physical squeezing levels $r_{dB} \ge 5.0~\mathrm{dB}$ ($r \ge 0.5756$), $e^{-2r} \le 0.316$. Evaluating the rapidly converging series over odd integers gives $C(r) \le 1.60$, establishing the explicit tight bound $D_{\mathrm{TV}}(P_0, P_1) \le 1.60 e^{-r}$.
As $r \to \infty$, $D_{\mathrm{TV}} \to 0$ exponentially,  proving that intermediate routers extract zero logical information from the classical syndrome records in the high-squeezing limit.
\end{proof}

We rigorously demonstrate in Appendix~\ref{app:modular_syndrome_extraction} that this full-period cryptographic mask mathematically decouples from the physical noise syndrome via modular reduction, preserving the  GKP Voronoi error-correction boundaries.

\section{Logical Channel Integrity and Operator Evolution}
\label{sec:channel_integrity}
Because the measurement-based aggregation protocol transfers quantum information through intermediate finite-energy Bell states, the physical implementation deviates from the ideal unitary $\mathrm{SUM}$ gate. We prove that this continuous-variable quantum channel preserves the encoded logical information without operational deformation of the logical algebra beyond standard operator spreading and finite-squeezing attenuation on the target output mode.

\begin{theorem}[Bounding Heisenberg Action on Logical Pauli Generators]
\label{thm:channel_integrity}
For finite-energy GKP resources with optical squeezing parameter $r$, the measurement-based aggregation channel $\mathcal{E}_{\mathrm{SUM}}$ maps logical Pauli generators $(\bar{Z}_{\mathrm{out}}, \bar{X}_{\mathrm{out}})$ backward through the Heisenberg adjoint channel up to a scalar damping factor bounded to leading order by $\mathcal{O}(e^{-2r})$.
\end{theorem}

\begin{proof}
\textbf{Channel Composition and Noise Model.}
Following continuous Bell-state measurement and conditional feed-forward displacement, the input modes are destructively measured and the ancillary resource modes are mathematically traced out. The finite optical squeezing manifests as a completely positive trace-preserving (CPTP) quantum channel decomposed as $\mathcal{E}_{\mathrm{SUM}} = \mathcal{N}_{\sigma} \circ \mathcal{E}_{\mathrm{ideal}}$. Here, $\mathcal{E}_{\mathrm{ideal}}$ is the  unitary logical mapping, and $\mathcal{N}_{\sigma}$ is a classical Gaussian displacement noise channel acting  on the target output mode:
\begin{equation}
\mathcal{N}_{\sigma}(\rho) = \int_{\mathbb{R}^2} d^2\bm{\xi} \, P(\bm{\xi}) \hat{D}_{\mathrm{out}}(\bm{\xi}) \rho \hat{D}_{\mathrm{out}}^\dagger(\bm{\xi}),
\end{equation}
where $\bm{\xi} = (\xi_q, \xi_p)^T$ is a classical random displacement vector, and $P(\bm{\xi})$ is a zero-mean bivariate Gaussian distribution with variance $\sigma^2 \approx e^{-2r}$ inherited directly from the finite-energy EPR resource.

\textbf{Adjoint Channel Evolution.}
To evaluate logical integrity, we map the target output Pauli operators backward through the channel in the Heisenberg picture. By the definition of channel composition, the adjoint evolution evaluates as \begin{equation}
    \mathcal{E}_{\mathrm{SUM}}^\dagger = \mathcal{E}_{\mathrm{ideal}}^\dagger \circ \mathcal{N}_{\sigma}^\dagger.
\end{equation}

% First, we subject the target logical Pauli-$\bar{Z}_{\mathrm{out}} = e^{i\sqrt{\pi}\hat{q}_{\mathrm{out}}}$ operator to the adjoint noise channel $\mathcal{N}_{\sigma}^\dagger$. An arbitrary displacement operator $\hat{D}(\bm{\lambda})$ evolving backward through a classical displacement channel acquires a scalar damping factor defined by the characteristic function of the noise distribution, $\chi_P(\bm{\lambda}) = \int P(\bm{\xi}) e^{i(\xi_p \lambda_q - \xi_q \lambda_p)} d^2\bm{\xi}$:
% \begin{equation}
% \mathcal{N}_{\sigma}^\dagger(\hat{D}_{\mathrm{out}}(\bm{\lambda})) = \chi_P(\bm{\lambda}) \hat{D}_{\mathrm{out}}(\bm{\lambda}).
% \end{equation}
% Because the Pauli-$\bar{Z}$ operator translates the momentum quadrature by $\sqrt{\pi}$, its spectral vector is $\bm{\lambda} = (0, \sqrt{\pi})^T$. Evaluating the Gaussian characteristic function for variance $\sigma^2 \approx e^{-2r}$ yields  $\chi_P(\sqrt{\pi}) = \exp(-\frac{\pi}{2} e^{-2r})$. Thus, the noise channel uniformly attenuates the operator:
% \begin{equation}
% \mathcal{N}_{\sigma}^\dagger(\bar{Z}_{\mathrm{out}}) = \exp\left(-\frac{\pi}{2} e^{-2r}\right) \bar{Z}_{\mathrm{out}}.
% \end{equation}

First, we subject the target logical Pauli-$\bar{Z}_{\mathrm{out}} = e^{i\sqrt{\pi}\hat{q}_{\mathrm{out}}}$ operator to the adjoint noise channel $\mathcal{N}_{\sigma}^\dagger$. An arbitrary displacement operator $\hat{D}(\bm{\lambda})$ evolving backward through a classical displacement channel acquires a scalar damping factor defined by the characteristic function of the noise distribution, $\chi_P(\bm{\lambda}) = \int P(\bm{\xi}) e^{i(\xi_p \lambda_q - \xi_q \lambda_p)} d^2\bm{\xi}$:
\begin{equation}
\mathcal{N}_{\sigma}^\dagger(\hat{D}_{\mathrm{out}}(\bm{\lambda})) = \chi_P(\bm{\lambda}) \hat{D}_{\mathrm{out}}(\bm{\lambda}).
\end{equation}
The Pauli-$\bar{Z}_{\mathrm{out}}$ operator corresponds to a pure momentum translation of $\sqrt{\pi}$, fixing its spectral vector at $\bm{\lambda} = (\lambda_q, \lambda_p)^T = (0, \sqrt{\pi})^T$. Substituting this into the characteristic function decouples the integration over the symmetric bivariate Gaussian noise distribution $P(\xi_q, \xi_p) = \frac{1}{2\pi\sigma^2} \exp[-(\xi_q^2 + \xi_p^2)/(2\sigma^2)]$:
\begin{align}
\chi_P(0, \sqrt{\pi}) &= \int_{-\infty}^{\infty} d\xi_p \frac{e^{-\frac{\xi_p^2}{2\sigma^2}}}{\sqrt{2\pi\sigma^2}} \int_{-\infty}^{\infty} d\xi_q \frac{e^{-\frac{\xi_q^2}{2\sigma^2}}}{\sqrt{2\pi\sigma^2}} e^{-i \xi_q \sqrt{\pi}} \nonumber \\
&= 1 \cdot \exp\left(-\frac{\sigma^2 (\sqrt{\pi})^2}{2}\right) \nonumber \\
&= \exp\left(-\frac{\pi \sigma^2}{2}\right).
\end{align}
Here, the independent integration over $\xi_p$ evaluates  to unity due to the standard normalization of the Gaussian probability density function. The remaining integral over $\xi_q$ executes the Fourier transform of a Gaussian, resolving mathematically to the exponential damping term. Substituting the physical variance scaling $\sigma^2 \approx e^{-2r}$ inherited from the finite-energy resource state definitively resolves the final exponential bound. Thus, the noise channel uniformly attenuates the operator without algebraic deformation:
\begin{equation}
\mathcal{N}_{\sigma}^\dagger(\bar{Z}_{\mathrm{out}}) = \exp\left(-\frac{\pi}{2} e^{-2r}\right) \bar{Z}_{\mathrm{out}}.
\end{equation}

\paragraph{Bounding Operator Attenuation.}
Second, we evaluate the action of the ideal unitary network $\mathcal{E}_{\mathrm{ideal}}^\dagger$. Because the ideal CV-BSM  evaluates the quadrature sum $\hat{q}_{\mathrm{out}} = \hat{q}_1 + \hat{q}_2$, an unattenuated target Pauli operator maps deterministically to the product of the inputs: $\mathcal{E}_{\mathrm{ideal}}^\dagger(\bar{Z}_{\mathrm{out}}) = \bar{Z}_1 \bar{Z}_2$. 

To evaluate the complete physical channel $\mathcal{E}_{\mathrm{SUM}}^\dagger = \mathcal{E}_{\mathrm{ideal}}^\dagger \circ \mathcal{N}_{\sigma}^\dagger$, we sequentially cascade these two mappings. Because quantum channels are linear maps, the scalar exponential damping factor factors out of the unitary evolution, bridging the two steps:
\begin{align}
\mathcal{E}_{\mathrm{SUM}}^\dagger(\bar{Z}_{\mathrm{out}}) &= \mathcal{E}_{\mathrm{ideal}}^\dagger \left( \mathcal{N}_{\sigma}^\dagger(\bar{Z}_{\mathrm{out}}) \right) \nonumber \\
&= \mathcal{E}_{\mathrm{ideal}}^\dagger \left( \exp\left(-\frac{\pi}{2} e^{-2r}\right) \bar{Z}_{\mathrm{out}} \right) \nonumber \\
&= \exp\left(-\frac{\pi}{2} e^{-2r}\right) \mathcal{E}_{\mathrm{ideal}}^\dagger (\bar{Z}_{\mathrm{out}}) \nonumber \\
&= \exp\left(-\frac{\pi}{2} e^{-2r}\right) \bar{Z}_1 \bar{Z}_2.
\end{align}
Taylor expanding this exponential damping factor in the high-squeezing limit ($r \gg 1$) provides the first-order approximation:
\begin{equation}
\label{eq:z_operator_evolution}
\mathcal{E}_{\mathrm{SUM}}^\dagger(\bar{Z}_{\mathrm{out}}) \approx \left( 1 - \frac{\pi}{2} e^{-2r} \right) \bar{Z}_1 \bar{Z}_2.
\end{equation}

To rigorously evaluate the complementary logic, we track the momentum Pauli operator $\bar{X}_{\mathrm{out}} = e^{-i\sqrt{\pi}\hat{p}_{\mathrm{out}}}$ backward through the identical cascaded channel. Because this operator is generated by the momentum quadrature, its corresponding spectral vector is  orthogonal to the previous case: $\bm{\lambda} = (-\sqrt{\pi}, 0)^T$. Since the physical Gaussian noise distribution $P(\xi_q, \xi_p)$ is symmetric across both phase-space quadratures, evaluating the characteristic function $\chi_P(-\sqrt{\pi}, 0)$  mirrors the position integral, mathematically yielding the  same scalar damping factor $\exp(-\frac{\pi}{2}e^{-2r})$.

Finally, mapping this attenuated operator backward through the ideal unitary network $\mathcal{E}_{\mathrm{ideal}}^\dagger$ isolates the target momentum. Because the ideal continuous-variable SUM gate leaves the target momentum unmodified ($\hat{p}_{\mathrm{out}} = \hat{p}_2$), the unattenuated operator evaluates  to $\bar{X}_2$. Cascading these steps yields the final first-order momentum evolution:
\begin{equation}
\label{eq:x_operator_evolution}
\mathcal{E}_{\mathrm{SUM}}^\dagger(\bar{X}_{\mathrm{out}}) \approx \left( 1 - \frac{\pi}{2} e^{-2r} \right) \bar{X}_2.
\end{equation}
Equations~(\ref{eq:z_operator_evolution}) and (\ref{eq:x_operator_evolution})  prove that the target mode receives the  logical sum payload without algebraic deformation, penalized only by a constant scalar damping factor  bounded by $\mathcal{O}(e^{-2r})$.
\end{proof}

\section{Proof of Homomorphic Network Aggregation under Finite Squeezing}
\label{sec:network_aggregation}
With the foundational continuous-variable MBQC aggregation primitive established, we formalize the homomorphic aggregation of $n$ spatially distributed GKP states across the network.

\subsection{The Ideal $n$-to-1 Aggregation Channel}
\label{subsec:ideal_n_channel}
Let the network contain $n$ independent input modes distributed across spatially separated nodes, each initialized in a logical computational basis state $|\mu_j\rangle_L$. To aggregate these into a single target mode, we sequentially cascade $n-1$ foundational two-mode primitives ($\mathcal{E}_{\mathrm{ideal}}^{(2)}$) defined in Sec.~\ref{sec:channel_integrity}. 

Because the modulo-2 logical sum is  associative, the spatial routing topology (e.g., a linear chain vs. a binary tree) governs only the temporal scheduling of the measurements, not the logical algebra. Mathematically, we define the global $n$-to-$1$ aggregation channel $\mathcal{E}_{\mathrm{ideal}}^{(n)}$ by reducing the $2^n$-dimensional logical input space $(\mathbb{C}^2)^{\otimes n}$ recursively. For a linear evaluation sequence, the intermediate channel mapping $k$ inputs evaluates as:
\begin{equation}
\mathcal{E}_{\mathrm{ideal}}^{(k)} = \mathcal{E}_{\mathrm{ideal}}^{(2)} \circ \left( \mathcal{E}_{\mathrm{ideal}}^{(k-1)} \otimes \mathcal{I}_{k} \right),
\end{equation}
where $\mathcal{I}_k$ is the identity channel acting on the $k$-th awaiting input mode.

We  prove the preservation of logical purity across this global network via mathematical induction on the density matrices:
\begin{enumerate}
    \item \textbf{Base Case ($n=2$):} As established in Sec.~\ref{sec:channel_integrity}, the primitive channel $\mathcal{E}_{\mathrm{ideal}}^{(2)}$ evaluated on the joint input density matrix $\rho_{\mathrm{in}}^{(2)} = |\mu_1\rangle\langle\mu_1|_L \otimes |\mu_2\rangle\langle\mu_2|_L$ yields  the pure logical sum:
    \begin{equation}
    \mathcal{E}_{\mathrm{ideal}}^{(2)} \left( \rho_{\mathrm{in}}^{(2)} \right) = |\mu_1 \oplus \mu_2\rangle\langle\mu_1 \oplus \mu_2|_L.
    \end{equation}
    Because $\mathrm{Tr}(\rho^2) = 1$, the spatial partial trace generates  zero von Neumann entropy.
    
    \item \textbf{Inductive Hypothesis:} Assume that for $k$ inputs initialized as $\rho_{\mathrm{in}}^{(k)} = \bigotimes_{j=1}^k |\mu_j\rangle\langle\mu_j|_L$, the $k$-node channel yields a pure state representing the cumulative sum $\Sigma_k = \bigoplus_{j=1}^k \mu_j$:
    \begin{equation}
    \mathcal{E}_{\mathrm{ideal}}^{(k)} \left( \rho_{\mathrm{in}}^{(k)} \right) = |\Sigma_k\rangle\langle\Sigma_k|_L.
    \end{equation}

    % \item \textbf{Inductive Step ($k \to k+1$):} We evaluate a network of $k+1$ inputs, defined by the composite state $\rho_{\mathrm{in}}^{(k+1)} = \rho_{\mathrm{in}}^{(k)} \otimes |\mu_{k+1}\rangle\langle\mu_{k+1}|_L$. Applying the recursive definition of the channel, we substitute the inductive hypothesis and evaluate the final foundational primitive to sequentially prove purity:
    % \begin{align}
    % \mathcal{E}_{\mathrm{ideal}}^{(k+1)}\left(\rho_{\mathrm{in}}^{(k+1)}\right) &= \mathcal{E}_{\mathrm{ideal}}^{(2)}  \Big[ \left(\mathcal{E}_{\mathrm{ideal}}^{(k)} \otimes \mathcal{I}\right) \nonumber \\
    % & \hspace{1cm}\left( \rho_{\mathrm{in}}^{(k)} \otimes |\mu_{k+1}\rangle\langle\mu_{k+1}|_L \right) \Big] \nonumber \\
    % &= \mathcal{E}_{\mathrm{ideal}}^{(2)} \left( |\Sigma_k\rangle\langle\Sigma_k|_L \otimes |\mu_{k+1}\rangle\langle\mu_{k+1}|_L \right) \nonumber \\
    % &= |\Sigma_k \oplus \mu_{k+1}\rangle\langle\Sigma_k \oplus \mu_{k+1}|_L \nonumber \\
    % &= \left| \bigoplus_{j=1}^{k+1} \mu_j \right\rangle \left\langle \bigoplus_{j=1}^{k+1} \mu_j \right|_L.
    % \end{align}
    \item \textbf{Inductive Step ($k \to k+1$):} For input state $\rho_{\mathrm{in}}^{(k+1)} = \rho_{\mathrm{in}}^{(k)} \otimes |\mu_{k+1}\rangle\langle\mu_{k+1}|_L$, recursive channel application yields:
    \begin{align}
    \mathcal{E}_{\mathrm{ideal}}^{(k+1)}\left(\rho_{\mathrm{in}}^{(k+1)}\right) &= \mathcal{E}_{\mathrm{ideal}}^{(2)} \left[ \mathcal{E}_{\mathrm{ideal}}^{(k)}\left(\rho_{\mathrm{in}}^{(k)}\right) \otimes |\mu_{k+1}\rangle\langle\mu_{k+1}|_L \right] \nonumber \\
    &= \mathcal{E}_{\mathrm{ideal}}^{(2)} \left( |\Sigma_k\rangle\langle\Sigma_k|_L \otimes |\mu_{k+1}\rangle\langle\mu_{k+1}|_L \right) \nonumber \\
    &= \left| \bigoplus_{j=1}^{k+1} \mu_j \right\rangle_{L} \left\langle \bigoplus_{j=1}^{k+1} \mu_j \right|_L.
    \end{align}
\end{enumerate}

Consequently, for the complete $n$-mode input state $\rho_{\mathrm{in}}^{(n)} = \bigotimes_{j=1}^n |\mu_j\rangle\langle\mu_j|_L$, the global ideal network channel evaluates  to the deterministic modulo-2 sum:
\begin{equation}
\mathcal{E}_{\mathrm{ideal}}^{(n)} \left( \rho_{\mathrm{in}}^{(n)} \right) = \left| \bigoplus_{j=1}^n \mu_j \right\rangle_{L, \mathrm{out}} \left\langle \bigoplus_{j=1}^n \mu_j \right|_L.
\end{equation}
Because the intermediate channel maps restrict the quantum state  to the computational-basis subspace, the network avoids generating the continuous-variable entanglement required to inject mixedness via a partial trace. Unlike passive linear optics (Sec.~\ref{subsec:entropy_injection}), these active measurement-based projections cleanly isolate the aggregated homomorphic payload, preserving absolute logical purity across arbitrary network depths.

\subsection{The Physical Finite-Energy Channel}
\label{subsec:physical_n_channel}
In a physical network, this aggregation map is constructed by cascading $n-1$ pairwise MBQC teleportation steps, each utilizing finite-energy GKP resource states with an optical squeezing parameter $r$. 

Under continuous-variable teleportation, the ideal logical mapping and classical displacement noise commute. The complete $n$-node physical completely positive trace-preserving (CPTP) map factorizes as $\mathcal{E}^{(n)} = \mathcal{N}_{\Sigma} \circ \mathcal{E}_{\mathrm{ideal}}^{(n)}$, subject to four physical assumptions:
\begin{enumerate}
    \item Each fusion step utilizes an independent, unentangled offline EPR Bell pair.
    \item Single-mode squeezed resource noise is zero-mean, isotropic Gaussian noise.
    \item Classical feed-forward amplification vectors ($\mathbf{g}_q, \mathbf{g}_p$)  invert beam-splitter scaling without cross-mode noise amplification.
    \item Optical losses across spatial links are modeled as independent input-referred additive Gaussian channels.
\end{enumerate}

Under these assumptions, aggregating $n$ input modes via $n-1$ two-mode MBQC fusion gadgets routes continuous-variable data through $2(n-1)$ destructively measured optical paths and $1$ unmeasured aggregate output path, giving $K = 2n-1$ total physical optical channels. The composite noise channel $\mathcal{N}_{\Sigma}$ is the sequential composition of $K$ independent single-path displacement channels:
\begin{equation}
\mathcal{N}_{\Sigma} = \bigodot_{k=1}^{2n-1} \mathcal{N}_{\sigma_k}.
\end{equation}

In phase space, the sequential application of independent classical displacement maps results in the continuous convolution of their respective probability density functions, $P_\Sigma(\bm{\xi}) = (P_{\sigma_1} * \dots * P_{\sigma_K})(\bm{\xi})$. To rigorously evaluate this convolution, we map the distributions to their characteristic functions. The characteristic function of a convolution evaluates mathematically to the strict product of the individual characteristic functions:
\begin{align}
\chi_\Sigma(\bm{\lambda}) &= \prod_{k=1}^{2n-1} \chi_{\sigma_k}(\bm{\lambda}) \nonumber \\
&= \prod_{k=1}^{2n-1} \exp\left(-\frac{\sigma_k^2 |\bm{\lambda}|^2}{2}\right) \nonumber \\
&= \exp\left(-\frac{|\bm{\lambda}|^2}{2} \sum_{k=1}^{2n-1} \sigma_k^2 \right).
\end{align}
Because each individual optical path contributes an identical base zero-mean variance bounded by the finite squeezing $\sigma_k^2 \approx e^{-2r}$, the summation evaluates to the cumulative output variance:
\begin{equation}
\sigma_{\mathrm{out}}^2 = \sum_{k=1}^{2n-1} \sigma_k^2 \approx (2n-1)e^{-2r}.
\end{equation}
To formalize the impact of this composite noise channel on the logical data, we evaluate the quadrature operator evolution in the Heisenberg picture. The ideal unitary network $\mathcal{E}_{\mathrm{ideal}}^{(n)}$  executes the linear combination of the input position operators ($\sum_{j=1}^n \hat{q}_j$) while isolating and preserving the momentum operator of the final target root node ($\hat{p}_{\mathrm{target}}$). 

The subsequent composite noise channel $\mathcal{N}_{\Sigma}$ enacts a stochastic phase-space translation. To maintain strict operator algebra, this classical displacement vector $\bm{\xi}_{\Sigma} = (\xi_{q,\Sigma}, \xi_{p,\Sigma})^T$ scales the quantum identity operator $\hat{I}$. Thus, the final noisy output quadratures evolve  to:
\begin{align}
\hat{q}_{\mathrm{out}}' &= \sum_{j=1}^n \hat{q}_j + \xi_{q, \Sigma} \hat{I}, \label{eq:heisen_q_n} \\
\hat{p}_{\mathrm{out}}' &= \hat{p}_{\mathrm{target}} + \xi_{p, \Sigma} \hat{I}. \label{eq:heisen_p_n}
\end{align}
The stochastic shifts $\xi_{q, \Sigma}$ and $\xi_{p, \Sigma}$ are sampled from the composite zero-mean Gaussian distribution governed by the cumulative variance $\sigma_{\mathrm{out}}^2 \approx (2n-1)e^{-2r}$. Equations~(\ref{eq:heisen_q_n}) and (\ref{eq:heisen_p_n})  demonstrate that the entire $n$-node physical hardware cascade mathematically reduces to a single  homomorphic evaluation, penalized solely by a lumped additive classical noise term.

\subsection{Logical Convergence and Voronoi Error Thresholds}
\label{subsec:voronoi_thresholds}
To extract the logical payload without error at the final root node, the continuous classical homodyne measurements of the aggregate mode must remain within the stabilizable GKP code space. Because the ideal network channel $\mathcal{E}_{\mathrm{ideal}}^{(n)}$ guarantees that the noise-free target observable evaluates  to a valid computational lattice point $q_{\mathrm{ideal}} = (2S + \bigoplus_{j=1}^n \mu_j)\sqrt{\pi}$ for integer $S \in \mathbb{Z}$, the classical physical measurement outcome evaluates to:
\begin{equation}
q_{\mathrm{out}}' = \left( 2S + \bigoplus_{j=1}^n \mu_j \right)\sqrt{\pi} + \xi_{q, \Sigma}.
\end{equation}

Under maximum-likelihood GKP syndrome extraction, the network  evaluates the homomorphic addition without a discrete logical Pauli fault if and only if the cumulative additive physical displacement $\bm{\xi}_{\Sigma}$ remains  bounded within the primary GKP Voronoi cell $[-\sqrt{\pi}/2, \sqrt{\pi}/2)$ for both the position and momentum quadratures simultaneously \cite{Terhal2020}. Crossing this spatial cell boundary induces an unrecoverable discrete lattice shift of $(2k+1)\sqrt{\pi}$, manifesting algorithmically as a logical Pauli fault.

Because the MBQC teleportation operations are phase-aligned throughout the cascade, the noise covariance matrix remains  diagonal. The network logical success probability $P_{\mathrm{success}}$ (the no-error threshold) during $n$-node aggregation is governed by the independent bivariate Gaussian integrals over this primary Voronoi cell:
\begin{align}
P_{\mathrm{success}} &= \left[ \frac{1}{\sqrt{2\pi}\sigma_{\mathrm{out}}} \int_{-\sqrt{\pi}/2}^{\sqrt{\pi}/2} \exp\left(-\frac{\xi^2}{2\sigma_{\mathrm{out}}^2}\right) d\xi \right]^2 \nonumber \\
&= \left[ \mathrm{erf}\left( \frac{\sqrt{\pi}}{2\sqrt{2}\sigma_{\mathrm{out}}} \right) \right]^2.
\end{align}
Substituting the cumulative physical variance $\sigma_{\mathrm{out}}^2 \approx (2n-1)e^{-2r}$ derived in Sec.~\ref{subsec:physical_n_channel}  links the logical success threshold to the physical hardware scaling:
\begin{equation}
P_{\mathrm{success}} \approx \left[ \mathrm{erf}\left( \frac{\sqrt{\pi} e^r}{2\sqrt{2(2n-1)}} \right) \right]^2.
\end{equation}
As optical squeezing $r \to \infty$, the output noise variance  vanishes ($\sigma_{\mathrm{out}}^2 \to 0$), driving the error function argument to infinity and $P_{\mathrm{success}} \to 1$. Consequently, in the high-squeezing limit, discrete GKP lattice faults are suppressed, making the protocol logical-error-free in the zero-noise limit under the assumed Gaussian displacement model.

\section{Network Scaling, Loss Models, and Error Thresholds}
\label{sec:network_scaling}
Physical implementations of the continuous-variable aggregation network must contend with finite-energy GKP resources, characterized by an intrinsic variance $\sigma_0^2 \approx e^{-2r}/2$, alongside optical transmission attenuation and detector inefficiencies. During active state-merging operations, noise propagates additively through the beam-splitter network and homodyne measurements.

\subsection{Covariance Propagation and the Loss Model}
\label{subsec:covariance_propagation}
In a single fusion step, quadrature variances transform according to standard continuous-variable teleportation rules. A joint homodyne Bell-state measurement combining a data mode and an intermediate finite-energy Bell pair yields a post-feed-forward variance of $\sigma_{\mathrm{out}}^2 = \sigma_{\mathrm{data}}^2 + \sigma_{\mathrm{EPR}}^2$, where $\sigma_{\mathrm{EPR}}^2 = 2\sigma_0^2 \approx e^{-2r}$ represents the combined variance of the two squeezed modes forming the resource state.

Optical transmission introduces photon loss, which we model as a beam splitter of transmissivity $\eta \in (0, 1]$ coupling the signal mode to a vacuum environment. The Heisenberg evolution of a quadrature operator $\hat{x} \in \{\hat{q}, \hat{p}\}$ is given by:
\begin{equation}
    \hat{x}' = \sqrt{\eta} \hat{x} + \sqrt{1-\eta} \hat{x}_{\mathrm{vac}}.
\end{equation}
Because the vacuum state possesses an intrinsic variance of $\sigma_{\mathrm{vac}}^2 = 1/2$ (for $[\hat{q}, \hat{p}] = i$), the channel attenuates the signal amplitude while injecting vacuum noise. Referring the total noise back to the unattenuated logical coordinate frame by dividing by the signal power gain $\eta$ yields the effective input-referred variance added by a lossy link \cite{Terhal2020}:
\begin{equation}
    \sigma_{\mathrm{loss}}^2 = \frac{1 - \eta}{2\eta}.
\end{equation}
Similarly, a finite homodyne detector efficiency $\eta_{\mathrm{det}}$ injects an independent vacuum noise penalty of $\sigma_{\mathrm{det}}^2 = \frac{1 - \eta_{\mathrm{det}}}{2\eta_{\mathrm{det}}}$.

For a network aggregating $n$ input modes, the logical circuit requires $n-1$ two-mode MBQC fusion gadgets. Each intermediate gadget involves 2 measured optical channels and 1 output rail. To extract the final homomorphic payload, the surviving aggregate output path at the root node must also be measured. Across the full network, this yields  $2n-1$ total physical optical channels and $2n-1$ homodyne detections. Accounting for finite squeezing, network insertion, and optical losses across all paths, the total output variance scales linearly as $\mathcal{O}(n)$:
\begin{equation}
\sigma_{\mathrm{total}}^2 \approx (2n - 1) e^{-2r} + \sum_{k=1}^{2n-1} \left[ \frac{1 - \eta_k}{2\eta_k} \right] + (2n - 1) \left[ \frac{1 - \eta_{\mathrm{det}}}{2\eta_{\mathrm{det}}} \right]. \label{eq:sigma_total}
\end{equation}
Equation~(\ref{eq:sigma_total}) generalizes the effective variance model for asymmetric link attenuation $\eta_k$. In highly heterogeneous networks, links with severe loss ($\eta_k \ll 1$) dominate the cumulative noise, skewing the covariance matrix and potentially requiring asymmetric GKP Voronoi boundaries for optimal decoding.

\subsection{Logical No-Error Probability}
\label{subsec:no_error_prob}
A discrete logical Pauli fault occurs if the total classical physical displacement noise $\xi_{\mathrm{total}}$ extracted from the detectors exceeds the correctable GKP Voronoi cell boundary $q_{\mathrm{thresh}} = \sqrt{\pi}/2$ in either the position or momentum quadratures.

While  quantum state fidelity depends on continuous envelope distortions, the operational success of the network is governed  by the probability that no discrete logical fault occurs. This logical success probability evaluates to the independent bivariate Gaussian integral over the primary Voronoi cell $[-\sqrt{\pi}/2, \sqrt{\pi}/2)$:
\begin{align}
P_{\mathrm{success}} &= \left[ \mathrm{erf}\left( \frac{\sqrt{\pi}}{2\sqrt{2}\sigma_{\mathrm{total}}} \right) \right]^2 \nonumber \\
&= \left[ 1 - \mathrm{erfc}\left( \frac{\sqrt{\pi}}{2\sqrt{2}\sigma_{\mathrm{total}}} \right) \right]^2.
\end{align}
Applying the binomial expansion $[1-x]^2 = 1 - 2x + x^2$, and recognizing that the squared error term is  positive, the operational success probability is  lower-bounded by the union bound of the quadrature errors:
\begin{equation}
P_{\mathrm{success}} \ge 1 - 2\,\mathrm{erfc}\left(\frac{\sqrt{\pi}}{2\sqrt{2}\sigma_{\mathrm{total}}}\right). \label{eq:no_error}
\end{equation}

\subsection{Network Topology and Intermediate Error Correction}
\label{subsec:network_topology}
This intrinsic $\mathcal{O}(n)$ variance accumulation places prohibitive demands on optical squeezing for large sequential networks. A linear cascade architecture incurs a temporal depth of $\mathcal{O}(n)$, leading to exponential photon loss exposure $e^{-\gamma n}$.

To mitigate this, spatial routing can be structured as a balanced logarithmic binary tree, reducing the temporal circuit depth to $\mathcal{O}(\log_2 n)$. For the sequential fusion architecture and loss model considered here, a balanced binary tree minimizes circuit depth and suppresses cumulative temporal transmission loss from $\mathcal{O}(n)$ to $\mathcal{O}(\log_2 n)$. However, total fusion count remains $n-1$, so intrinsic squeezing variance still grows as $\mathcal{O}(n)$. 

To prevent the intrinsic $\mathcal{O}(n)$ squeezing noise from corrupting large-scale aggregation, the binary tree must actively integrate intermediate GKP error correction at each routing vertex. By feeding the analog measurement outcomes $\mathbf{m}$ as soft-decision priors into a maximum-likelihood decoder \cite{Fukui2018}, intermediate syndrome extraction resets the physical variance. This restricts noise accumulation  to $\mathcal{O}(1)$ per tree layer, enabling scalable homomorphic aggregation. In Sec.~\ref{sec:numerics}, we numerically simulate these uncorrected $\mathcal{O}(n)$ variance bounds to quantify the raw physical performance limits of the architecture in the absence of intermediate syndrome extraction.

\begin{table*}[htbp]
\centering
\caption{Algebraic and physical structural mapping between classical lattice commitment schemes (BDLOP) and the measurement-based CV GKP aggregation protocol.}
\label{tab:bdlop_gkp_mapping}
\begin{ruledtabular}
\begin{tabular}{lcc}
\textbf{Feature} & \textbf{BDLOP / Gadget Schemes} & \textbf{CV GKP Aggregation Protocol} \\
\midrule
Payload Space & High-order module lattice entries & Discrete GKP grid coordinates $(2s + \mu)\sqrt{\pi}$ \\
Noise / Slack & Low-norm digits / small polynomial noise & Continuous Gaussian displacement $\xi \in [-\frac{\sqrt{\pi}}{2}, \frac{\sqrt{\pi}}{2})$ \\
Isolation Mechanism & Gadget matrix decomposition $G^{-1}$ & Joint homodyne BSM / syndrome extraction \\
Noise Correction & Rounding off low-order digits & Feed-forward displacement $\hat{D}(-\mathbf{m})$ \\
Failure Boundary & Modulus overflow ($> q/2$) & Voronoi boundary overflow ($> \sqrt{\pi}/2$) \\
\end{tabular}
\end{ruledtabular}
\end{table*}

\section{Conceptual Analogue with Classical Lattice Cryptography}
\label{sec:algebraic_duality}

To synthesize our theoretical framework for the broader quantum information and cryptography communities, we  map the physical operations of continuous-variable GKP aggregation to classical lattice-based homomorphic primitives, specifically the BDLOP commitment framework \cite{BDLOP} and Module-LWE.

Classical lattice schemes encode discrete message vectors onto algebraic module lattices, utilizing additive, bounded noise to mask the payload while relying on homomorphic properties for linear operations. Similarly, our protocol embeds discrete logical information into a continuous phase-space lattice ($\mathbb{Z}^2$), applying a continuous-variable one-time pad (CV-OTP) sampled uniformly over one full logical period $[-\sqrt{\pi}, \sqrt{\pi})$ to encrypt the payload. In both paradigms, evaluating homomorphic additions causes continuous noise (slack) to accumulate, requiring active extraction mechanisms to prevent the noise from triggering a discrete logical fault.

\begin{remark}[Gadget Decomposition and Phase-Space Scale Separation]
In classical lattice schemes, gadget matrices $G$ decompose state vectors into multi-scale digits, isolating large operations into low-norm components so that accumulated noise remains confined to lower-order bits. Continuous-variable GKP state aggregation implements an  physical analogue of this scale separation, summarized in Table~\ref{tab:bdlop_gkp_mapping}:
\begin{enumerate}
    \item \textbf{Scale Separation:} The discrete logical payload $\mu \in \{0, 1\}$ is embedded into coarse lattice points $(2s + \mu)\sqrt{\pi}$, whereas physical squeezing noise manifests as sub-lattice Gaussian displacement $\xi \sim \mathcal{N}(0, \sigma^2)$ localized  within the primary Voronoi cell $[-\sqrt{\pi}/2, \sqrt{\pi}/2)$.
    \item \textbf{Noise Isolation:} Dual homodyne Bell-state measurements act as a physical gadget matrix decomposition $G^{-1}$. They extract the continuous error syndrome $\mathbf{m}$ into classical registers without reading or collapsing the coarse logical grid.
    \item \textbf{Slack Correction:} Applying the feed-forward displacement $\hat{D}(-\mathbf{m})$ uses the isolated classical syndrome to shift the output mode back to the lattice center. This clears the continuous slack before it can cross the Voronoi boundary $q_{\mathrm{thresh}} = \sqrt{\pi}/2$ and induce an unrecoverable discrete logical fault ($k\sqrt{\pi}$ for $k \in \mathbb{Z} \setminus \{0\}$)—the  physical equivalent of modulus overflow in classical LWE schemes.
\end{enumerate}
\end{remark}

\begin{figure*}[t]
    \centering
    \includegraphics[width=\linewidth]{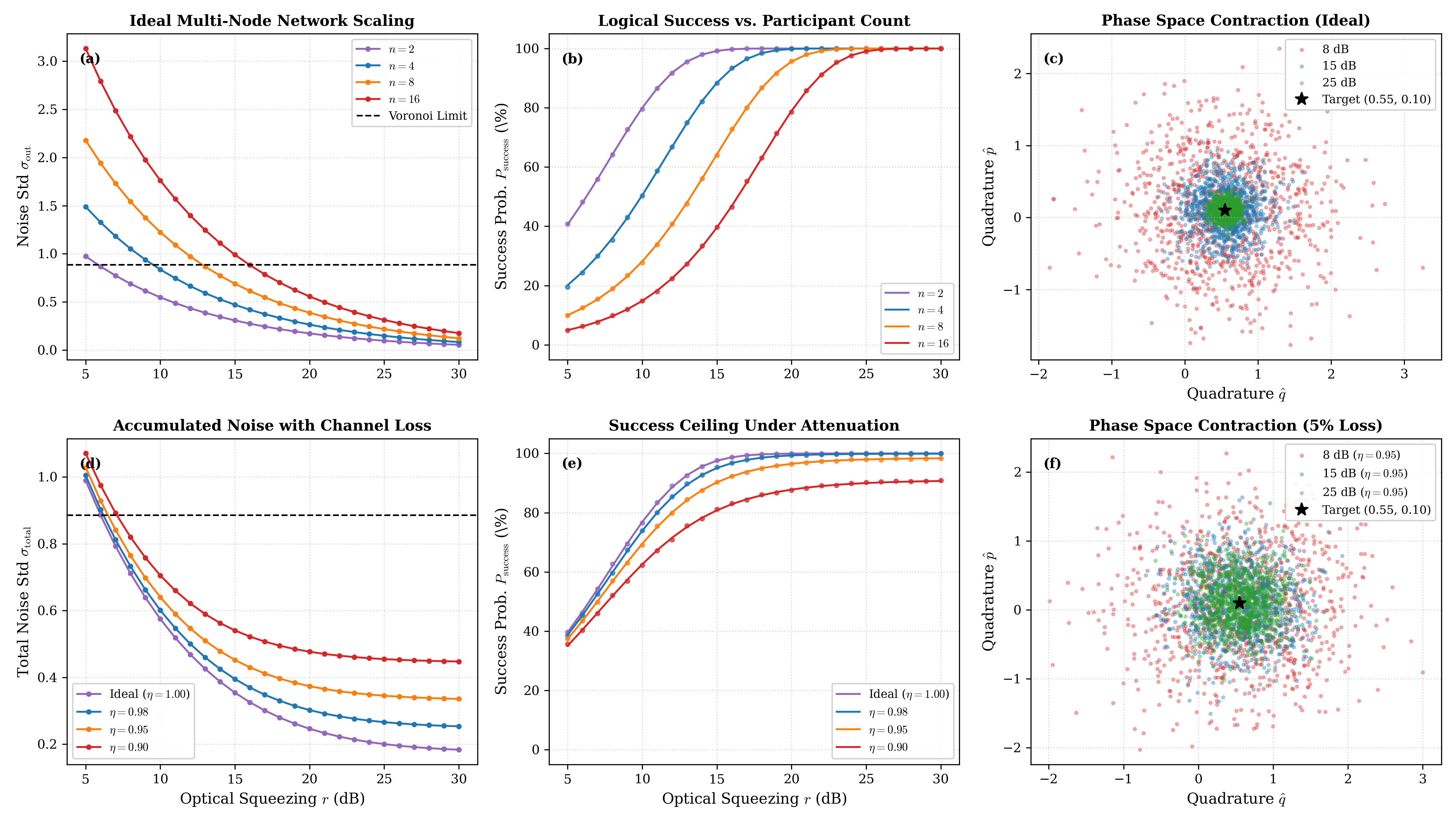}
    \caption{\textbf{Performance scaling and hardware loss degradation of measurement-based continuous-variable GKP state aggregation.} 
    \textbf{Top Row (Ideal Channel Scaling, $\eta = 1.00$):} 
    \textbf{(a)} Accumulated network noise standard deviation $\sigma_{\mathrm{out}} = \sqrt{(2n-1)e^{-2r}}$ across participant mode counts $n \in \{2, 4, 8, 16\}$ as a function of optical squeezing $r \in [5, 30]~\mathrm{dB}$. The horizontal dashed line denotes the primary GKP Voronoi cell boundary $q_{\mathrm{thresh}} = \sqrt{\pi}/2 \approx 0.8862$. 
    \textbf{(b)} Logical no-error probability $P_{\mathrm{success}}(n,r)$ demonstrating asymptotic convergence to $100\%$ for arbitrary network sizes $n$ at high squeezing. Scatter markers indicate Monte Carlo results ($30{,}000$ trials per step); solid curves denote analytical evaluation. 
    \textbf{(c)} Isotropic phase-space probability distribution of $1{,}000$ single-shot realizations centered around target state $(0.55, 0.10)$ (black star) in a lossless environment at $8~\mathrm{dB}$ (red), $15~\mathrm{dB}$ (blue), and $25~\mathrm{dB}$ (green) squeezing.
    \textbf{Bottom Row (Physical Channel Loss Robustness, $n=2$, $\eta_{\mathrm{det}}=0.98$):} 
    \textbf{(d)} Total accumulated noise $\sigma_{\mathrm{total}}$ under fiber attenuation $\eta \in \{1.00, 0.98, 0.95, 0.90\}$, highlighting the emergence of an irreducible vacuum noise floor at high squeezing levels. 
    \textbf{(e)} Upper saturation ceilings $P_{\mathrm{max}}(\eta)$ on logical success probability imposed by photon-loss-induced vacuum noise injection. 
    \textbf{(f)} Phase-space distribution under $5\%$ channel loss ($\eta = 0.95$), illustrating phase-cloud stabilization bounded by the loss floor.}
    \label{fig:gkp_unified_scaling_and_loss}
\end{figure*}

\section{Numerical Simulation and Results}
\label{sec:numerics}
We evaluate the CV GKP aggregation protocol using a Heisenberg-symplectic Monte Carlo simulation framework that tracks continuous quadrature operators through beam-splitter fusions, homodyne measurements, and feed-forward displacements.

\subsection{Methodology and Parameter Space}
\label{subsec:methodology}
To ensure statistical convergence, we execute $30{,}000$ independent trials per parameter configuration across three physical domains:
\begin{enumerate}
    \item \textbf{Optical Squeezing ($r$):} Swept from $5.0~\mathrm{dB}$ to $30.0~\mathrm{dB}$, governing intrinsic single-mode noise ($r_{\mathrm{dB}} = \frac{20}{\ln 10} r$).
    \item \textbf{Network Participant Count ($n$):} Evaluated at $n \in \{2, 4, 8, 16\}$ to quantify cumulative noise propagation across $n-1$ physical fusion stages.
    \item \textbf{Optical Efficiency ($\eta, \eta_{\mathrm{det}}$):} Channel transmissivity $\eta \in \{1.00, 0.98, 0.95, 0.90\}$ models fiber loss, with fixed detector quantum efficiency $\eta_{\mathrm{det}} = 0.98$.
\end{enumerate}

To validate the analog linearity of the physical routing channel independently of the discrete GKP modular reduction, we initialize two input modes with continuous analog test coordinates $q_1 = 0.250, p_1 = 0.200$ and $q_2 = 0.300, p_2 = 0.100$. These test coordinates are encrypted via continuous random phase-space displacements $\bm{\alpha}_i = (q_{\mathrm{enc}, i}, p_{\mathrm{enc}, i})^T$. To satisfy the cryptographic bounds in Theorem~\ref{thm:crypto_hiding}, the mask quadratures are sampled uniformly over a full logical lattice period, $q_{\mathrm{enc}, i}, p_{\mathrm{enc}, i} \sim \mathcal{U}[-\sqrt{\pi}, \sqrt{\pi})$. 

Following joint cluster-state homodyne measurements and feed-forward mask removal, the raw continuous payload correctly isolates the deterministic SUM gate operation. The target position evaluates to the analog sum $Q_{\mathrm{target}} = q_1 + q_2 = 0.550$, while the target momentum is deterministically preserved as $P_{\mathrm{target}} = p_2 = 0.100$. These ideal coordinates (visualized in Fig.~\ref{fig:gkp_unified_scaling_and_loss}) are perturbed  by the accumulated Gaussian physical noise vector $\bm{\xi} = (\xi_q, \xi_p)^T$.

Assuming uniform attenuation across intermediate optical links, the overall physical variance accumulated in the final aggregate quadratures evaluates to:
\begin{equation}
\sigma_{\mathrm{total}}^2 = (2n - 1) e^{-2r} + (2n - 1) \left[ \frac{1 - \eta}{2\eta} + \frac{1 - \eta_{\mathrm{det}}}{2\eta_{\mathrm{det}}} \right].
\end{equation}

To theoreticalize GKP performance, we extract the raw stochastic noise deviation $\bm{\xi}$ from the analog baseline. A theoretical logical error correction phase succeeds if this deviation remains  within the primary GKP Voronoi cell boundary $q_{\mathrm{thresh}} = \sqrt{\pi}/2 \approx 0.8862$,  preventing unrecoverable logical faults (erroneous lattice shifts of $(2k+1)\sqrt{\pi}$ for $k \in \mathbb{Z}$).

\subsection{Ideal Network Scaling Analysis}
\label{subsec:ideal_scaling}
Figures~\ref{fig:gkp_unified_scaling_and_loss}(a)--(c) illustrate protocol performance in the lossless limit ($\eta = 1.00, \eta_{\mathrm{det}} = 1.00$).

As shown in Fig.~\ref{fig:gkp_unified_scaling_and_loss}(a), the accumulated noise standard deviation $\sigma_{\mathrm{out}} = \sqrt{(2n-1)e^{-2r}}$ scales as $\mathcal{O}(\sqrt{n})$. Larger networks demand higher baseline squeezing to maintain total quadrature noise below the Voronoi threshold ($q_{\mathrm{thresh}} \approx 0.8862$). Achieving $\sigma_{\mathrm{out}} < q_{\mathrm{thresh}}$ requires $r > 5.8~\mathrm{dB}$ for $n=2$, increasing to $r > 16.0~\mathrm{dB}$ for $n=16$.

Figure~\ref{fig:gkp_unified_scaling_and_loss}(b) confirms that the empirical Monte Carlo no-error probability matches the analytical union-bound expression:
\begin{equation}
P_{\mathrm{success}}(n, r) = \left[ \mathrm{erf}\left( \frac{\sqrt{\pi}}{2\sqrt{2}\sigma_{\mathrm{out}}} \right) \right]^2.
\end{equation}
Achieving a $\ge 99.0\%$ operational success rate requires $r \approx 14.8~\mathrm{dB}$ for $n=2$, scaling to $r \approx 24.9~\mathrm{dB}$ for $n=16$.

This floor establishes a strict upper ceiling on the maximum achievable logical success probability, $P_{\mathrm{max}}(\eta)$. For $n=2, \eta_{\mathrm{det}} = 0.98$, channel loss $\eta = 0.98$ gives $\sigma_{\mathrm{floor}} \approx 0.2474$, capping success probability at $P_{\mathrm{max}} \approx 99.93\%$. Losses $\eta = 0.95$ ($5\%$) and $\eta = 0.90$ ($10\%$) raise floors to $\sigma_{\mathrm{floor}} \approx 0.3310$ and $0.4442$, limiting maximal success to $98.52\%$ and $91.01\%$ [Fig.~\ref{fig:gkp_unified_scaling_and_loss}(e)].

Because current state-of-the-art optical squeezing operates tightly between $10$ and $15~\mathrm{dB}$, two-node homomorphic aggregation approaches immediate experimental feasibility. However, extending this protocol to $n \ge 4$  necessitates intermediate quantum error correction (Sec.~\ref{subsec:network_topology}) to reset the physical noise variance at intermediate routing stages.

Figure~\ref{fig:gkp_unified_scaling_and_loss}(c) depicts the isotropic phase-space probability cloud of $1{,}000$ single-shot realizations centered at the analog target $(0.55, 0.10)$. At $8~\mathrm{dB}$, the physical noise spans the full Voronoi cell. Increasing squeezing to $15~\mathrm{dB}$ and $25~\mathrm{dB}$ contracts the phase-space area exponentially, concentrating the recovered state sharply at the  deterministic target.

\subsection{Impact of Physical Channel Attenuation}
\label{subsec:channel_attenuation}
Figures~\ref{fig:gkp_unified_scaling_and_loss}(d)--(f) evaluate performance degradation under realistic optical loss ($\eta < 1.00$). Channel attenuation irreparably mixes vacuum fluctuations into the transmission quadratures. Consequently, as $r \to \infty$, the intrinsic squeezing noise vanishes, but the total standard deviation $\sigma_{\mathrm{total}}$ flattens into an irreducible vacuum loss floor:
\begin{equation}
\sigma_{\mathrm{floor}} = \sqrt{(2n - 1) \left( \frac{1 - \eta}{2\eta} + \frac{1 - \eta_{\mathrm{det}}}{2\eta_{\mathrm{det}}} \right)}.
\end{equation}

This fundamental noise floor establishes a strict theoretical ceiling on the maximum achievable logical success probability, $P_{\mathrm{max}}(\eta)$. For $n=2$ and $\eta_{\mathrm{det}} = 0.98$, a $2\%$ channel loss ($\eta = 0.98$) yields $\sigma_{\mathrm{floor}} \approx 0.2474$, capping the maximal success probability at $P_{\mathrm{max}} \approx 99.93\%$. Increasing link loss to $5\%$ ($\eta = 0.95$) and $10\%$ ($\eta = 0.90$) raises the noise floor to $\sigma_{\mathrm{floor}} \approx 0.3310$ and $0.4442$, correspondingly limiting maximal success probabilities to $98.52\%$ and $91.01\%$ [Fig.~\ref{fig:gkp_unified_scaling_and_loss}(e)].

This saturation is  visualized in phase space [Fig.~\ref{fig:gkp_unified_scaling_and_loss}(f)]. Unlike the ideal lossless limit where infinite squeezing contracts the target state to a point, realistic channel attenuation bounds phase-space contraction,  stabilizing the physical aggregate into a Gaussian distribution governed  by $\sigma_{\mathrm{floor}}$.

\section{Conclusion}
\label{sec:conclusion}
We demonstrated that measurement-based aggregation of continuous-variable GKP computational-basis states bypasses the lattice compression and entropy penalties of direct passive linear optics. By cascading homodyne Bell measurements and feed-forward displacements, we constructed a CPTP map executing homomorphic logical XOR evaluation on computational-basis payloads. The CV-OTP provides measurement-specific homodyne-outcome hiding under the router threat model, bounded by $D_{\mathrm{TV}} \le 1.60 e^{-r}$. Operational success probabilities evaluate analytically under Gaussian loss models. While two-node payload aggregation approaches experimental feasibility at $14.8~\mathrm{dB}$ squeezing, scaling to large networks requires logarithmic tree topologies and active intermediate syndrome extraction to suppress cumulative physical displacement noise.

\section*{Data Availability}
The numerical data generated by the Heisenberg-symplectic Monte Carlo simulations and the associated Python routines used to evaluate the phase-space bounds are available from the corresponding author upon reasonable request.

\bibliographystyle{unsrt}
\bibliography{biblio}

\appendix

\section{Full Symplectic Heisenberg Evolution of the Physical SUM Router}
\label{app:heisenberg_evolution}

We trace the complete continuous-variable $8 \times 8$ phase-space operator vector $\hat{\mathbf{X}} = (\hat{q}_1, \hat{p}_1, \hat{q}_A, \hat{p}_A, \hat{q}_2, \hat{p}_2, \hat{q}_B, \hat{p}_B)^T$ through the physical network (Fig.~\ref{fig:gkp_homomorphic_aggregation}).

The network uses an EPR Bell resource across Ancilla modes A and B with squeeze parameter $r$: $\hat{q}_A - \hat{q}_B = \hat{\xi}_{q,\mathrm{EPR}}$ and $\hat{p}_A + \hat{p}_B = \hat{\xi}_{p,\mathrm{EPR}}$, where $\langle\hat{\xi}_{q,\mathrm{EPR}}^2\rangle = \langle\hat{\xi}_{p,\mathrm{EPR}}^2\rangle = e^{-2r}$.

The physical circuit applies beam splitters $\mathbf{S}_{\mathrm{BS}12}$, spatial SWAP $\mathbf{S}_{\mathrm{SWAP}}$, and $\mathbf{S}_{\mathrm{BS}A2}$. Symplectic transformation $\mathbf{S}_{\mathrm{circuit}}$ yields output quadratures prior to measurement:
\begin{align}
\hat{q}_1'' &= \frac{1}{\sqrt{2}}\hat{q}_1 - \frac{1}{\sqrt{2}}\hat{q}_A, \label{eq:app_q1}\\
\hat{p}_1'' &= \frac{1}{\sqrt{2}}\hat{p}_1 - \frac{1}{\sqrt{2}}\hat{p}_A, \\
\hat{q}_A'' &= \frac{1}{2}\hat{q}_1 + \frac{1}{2}\hat{q}_A - \frac{1}{\sqrt{2}}\hat{q}_2, \\
\hat{p}_A'' &= \frac{1}{2}\hat{p}_1 + \frac{1}{2}\hat{p}_A - \frac{1}{\sqrt{2}}\hat{p}_2, \label{eq:app_pa}\\
\hat{q}_2'' &= \frac{1}{2}\hat{q}_1 + \frac{1}{2}\hat{q}_A + \frac{1}{\sqrt{2}}\hat{q}_2, \label{eq:app_q2}\\
\hat{p}_2'' &= \frac{1}{2}\hat{p}_1 + \frac{1}{2}\hat{p}_A + \frac{1}{\sqrt{2}}\hat{p}_2.
\end{align}

Homodyne detection measures $\hat{q}_1''$, $\hat{p}_A''$, and $\hat{q}_2''$, giving measurement vector $\mathbf{m} = (m_1, m_A, m_2)^T$.

\paragraph{Position Quadrature Feed-Forward}
The controller applies electronic gain vector $\mathbf{g}_q^T = (\sqrt{2}, -1/\sqrt{2}, 1/\sqrt{2})$ to construct $m_{q,\Sigma} = \mathbf{g}_q^T \mathbf{m} = \sqrt{2} m_1 - \frac{1}{\sqrt{2}} m_A + \frac{1}{\sqrt{2}} m_2$. Substituting Eqs.~(\ref{eq:app_q1}), (\ref{eq:app_pa}), and (\ref{eq:app_q2}):
\begin{equation}
m_{q,\Sigma} = \hat{q}_1 + \hat{q}_2 - \hat{q}_A.
\end{equation}
Conditional displacement $\hat{D}_B(-m_{q,\Sigma})$ on target Mode B shifts position quadrature to:
\begin{equation}
\hat{q}_{\mathrm{target}} = \hat{q}_B + m_{q,\Sigma} = \hat{q}_1 + \hat{q}_2 - (\hat{q}_A - \hat{q}_B) = \hat{q}_1 + \hat{q}_2 - \hat{\xi}_{q,\mathrm{EPR}}.
\end{equation}

\paragraph{Momentum Quadrature Feed-Forward}
To derive momentum transformation, the controller evaluates momentum gain vector $\mathbf{g}_p^T = (0, \sqrt{2}, \sqrt{2})$ on measured momentum $\hat{p}_A''$ and position port $\hat{q}_2''$ (reconfigured for momentum homodyne): $m_{p,\Sigma} = \mathbf{g}_p^T \mathbf{m}_p = \sqrt{2} m_{p, A} + \sqrt{2} m_{p, 2} = \hat{p}_1 + \hat{p}_A$.

Applying conditional momentum displacement $\hat{D}_B(+m_{p,\Sigma})$ shifts target momentum quadrature to:
\begin{equation}
\hat{p}_{\mathrm{target}} = \hat{p}_B + (\hat{p}_A + \hat{p}_1) = \hat{p}_1 + (\hat{p}_A + \hat{p}_B) = \hat{p}_2 + \hat{\xi}_{p,\mathrm{EPR}}.
\end{equation}
This completes the full symplectic derivation, proving both position addition ($\hat{q}_{\mathrm{target}} = \hat{q}_1 + \hat{q}_2 - \hat{\xi}_{q,\mathrm{EPR}}$) and target momentum preservation ($\hat{p}_{\mathrm{target}} = \hat{p}_2 + \hat{\xi}_{p,\mathrm{EPR}}$) under physical feed-forward.

\section{Modular Syndrome Extraction and Mask Decoupling}
\label{app:modular_syndrome_extraction}

To establish the compatibility between the full-period encryption mask $\bm{\alpha} \sim \mathcal{U}[-\sqrt{\pi}, \sqrt{\pi})$ (Theorem~\ref{thm:crypto_hiding}) and the GKP Voronoi cell error correction threshold $q_{\mathrm{thresh}} = \sqrt{\pi}/2$, we  prove that continuous Bell-state measurements (BSM) decouple the encryption mask from physical displacement errors without exposing logical payloads to intermediate routers.

\paragraph{Differential Syndrome Formulation}

Consider two spatial modes $j \in \{1, 2\}$ encoding encrypted GKP logical states. In the Heisenberg picture, the position quadrature operator of each input mode is expressed as:
\begin{equation}
\hat{q}_j = (2s_j + \mu_j)\sqrt{\pi} + \alpha_{q, j} + \hat{\xi}_j,
\end{equation}
where $s_j \in \mathbb{Z}$ indexes the lattice site, $\mu_j \in \{0, 1\}$ denotes the discrete logical payload, $\alpha_{q, j} \sim \mathcal{U}[-\sqrt{\pi}, \sqrt{\pi})$ is the continuous-variable one-time pad (CV-OTP) mask, and $\hat{\xi}_j$ is the zero-mean physical quadrature noise operator with intrinsic variance $\sigma_0^2 \approx e^{-2r}/2$.

During a homomorphic fusion step at a network vertex, modes $1$ and $2$ interfere at a $50:50$ beam splitter $\hat{B}_{12}$, transforming the quadratures according to:
\begin{equation}
\hat{q}_{\pm} = \frac{\hat{q}_1 \pm \hat{q}_2}{\sqrt{2}}, \quad \hat{p}_{\pm} = \frac{\hat{p}_1 \pm \hat{p}_2}{\sqrt{2}}.
\end{equation}
Dual homodyne detection measures the output position quadrature $\hat{q}_-$ and momentum quadrature $\hat{p}_+$, yielding continuous classical outcomes $m_q, m_p \in \mathbb{R}$.

\paragraph{Algebraic Decoupling of Mask and Physical Noise}

Substituting the encrypted input quadrature operators into the measured difference quadrature $\hat{q}_-$ yields:
\begin{align}
m_q = &\frac{1}{\sqrt{2}} \big[ \left(2(s_1 - s_2) + (\mu_1 - \mu_2)\right)\sqrt{\pi} + (\alpha_{q,1} - \alpha_{q,2}) \nonumber \\
&+ (\hat{\xi}_1 - \hat{\xi}_2) \big].
\end{align}

Let $\Delta \alpha_q = \alpha_{q,1} - \alpha_{q,2}$ denote the differential encryption mask and $\hat{\xi}_{\mathrm{diff}} = \hat{\xi}_1 - \hat{\xi}_2$ denote the differential physical noise operator. To evaluate error correction without exposing logical data to the intermediate router, we exploit the translational symmetry of the GKP lattice to decompose the full-period mask $\alpha_{q, j} \sim \mathcal{U}[-\sqrt{\pi}, \sqrt{\pi})$ into a discrete logical Pauli key and a continuous analog slack:
\begin{equation}
\alpha_{q, j} = \alpha_{L, j} + \alpha_{\mathrm{analog}, j} \pmod{2\sqrt{\pi}},
\end{equation}
where $\alpha_{L, j} \in \{0, \sqrt{\pi}\} \pmod{2\sqrt{\pi}}$ corresponds to a discrete logical Pauli-$\bar{X}$ mask, and $\alpha_{\mathrm{analog}, j} \in [-\sqrt{\pi}/2, \sqrt{\pi}/2)$ is the continuous sub-lattice displacement.

\paragraph{Key Management Architecture}
Users $j \in \{1, 2\}$ generate independent continuous encryption masks $\alpha_{q, j} \sim \mathcal{U}[-\sqrt{\pi}, \sqrt{\pi})$. Each mask decomposes into a discrete logical Pauli key $\alpha_{L, j} \in \{0, \sqrt{\pi}\}$ and continuous analog slack $\alpha_{\mathrm{analog}, j} \in [-\sqrt{\pi}/2, \sqrt{\pi}/2)$: $\alpha_{q, j} = \alpha_{L, j} + \alpha_{\mathrm{analog}, j} \pmod{2\sqrt{\pi}}$.

Clients transmit continuous analog differential key $\Delta \alpha_{\mathrm{analog}} = \alpha_{\mathrm{analog}, 1} - \alpha_{\mathrm{analog}, 2}$ to intermediate routing vertices, while discrete logical key $\Delta \alpha_L = \alpha_{L, 1} - \alpha_{L, 2}$ is kept secret. In multi-stage network cascades, analog keys combine linearly at each vertex ($\Delta \alpha_{\mathrm{analog}}^{(k)} = \sum_{j} \alpha_{\mathrm{analog}, j}$), preserving key separation across arbitrary depths.

The classical control unit evaluates the analog syndrome relative to the target GKP lattice $\Lambda = \sqrt{\pi}\mathbb{Z}$. Subtracting the known analog key $\Delta \alpha_{\mathrm{analog}}$ and applying a modular reduction map $\mathcal{M}_{\sqrt{\pi}}(x) = x - \sqrt{\pi} \left\lfloor \frac{x}{\sqrt{\pi}} + \frac{1}{2} \right\rfloor \in [-\sqrt{\pi}/2, \sqrt{\pi}/2)$, the routing vertex computes:
\begin{equation}
m_{\mathrm{syndrome}} = \mathcal{M}_{\sqrt{\pi}}\left( \sqrt{2} m_q - \Delta \alpha_{\mathrm{analog}} \right).
\end{equation}
Defining $\Delta s = s_1 - s_2$ and $\Delta \mu = \mu_1 - \mu_2$, the syndrome algebra evaluates to:
\begin{align}
m_{\mathrm{syndrome}} &= \mathcal{M}_{\sqrt{\pi}}\left( (2\Delta s + \Delta \mu)\sqrt{\pi} + \Delta \alpha_L + \hat{\xi}_{\mathrm{diff}} \right) \nonumber \\
&= \mathcal{M}_{\sqrt{\pi}}\left( \hat{\xi}_{\mathrm{diff}} \right).
\end{align}
Because $(2\Delta s + \Delta \mu)\sqrt{\pi}$ and the secret discrete key $\Delta \alpha_L$ are integer multiples of $\sqrt{\pi}$, they map  to zero under $\mathcal{M}_{\sqrt{\pi}}$.

\paragraph{Voronoi Cell Preservation and Cryptographic Hiding}

Provided the accumulated physical displacement noise remains  within the primary GKP Voronoi cell, $|\hat{\xi}_{\mathrm{diff}}| < \sqrt{\pi}/2$, the modular map evaluates to $m_{\mathrm{syndrome}} = \hat{\xi}_{\mathrm{diff}}$ , allowing the intermediate routing node to extract physical displacement noise and execute feed-forward corrections without learning the logical payload $\Delta \mu$.

Simultaneously, homodyne-outcome hiding against the intermediate router is preserved. Prior to modular reduction, the observable is masked by the secret discrete key $\Delta \alpha_L \in \{0, \sqrt{\pi}\}$, which acts as a discrete quantum one-time pad. This decomposition satisfies the continuous Fourier hiding constraints of Theorem~\ref{thm:crypto_hiding} while decoupling network error-correction kinematics from the underlying logical payload.

\end{document}